\pdfoutput=1
\documentclass[11pt]{article} 
\usepackage{fullpage}
\usepackage[T1]{fontenc}
 \usepackage{authblk}
\usepackage{amssymb,amsthm,amsmath,amssymb,wrapfig,dsfont}
\usepackage{thmtools,thm-restate}
\usepackage{url}
\usepackage{color}
  \usepackage{sectsty}
   \usepackage{stmaryrd}
\usepackage{bbold}
\usepackage{tikz}
 \usetikzlibrary{shapes.geometric, arrows, fit}
 \usetikzlibrary{decorations.pathreplacing}

\usepackage{caption} 
\usepackage{soul}
\usepackage{algorithm}
\usepackage{algorithmic}
\usepackage{pgfplots}   
\pgfplotsset{width=10cm,compat=newest}
\usepgfplotslibrary{units}
\usetikzlibrary{spy,backgrounds}
\pgfplotsset{width=10cm,compat=1.9}
\def\eps{\varepsilon}
\usepackage{newfloat}
\usepackage{listings}
\DeclareCaptionStyle{ruled}
{labelfont=normalfont,labelsep=colon,strut=off} 

\newcommand{\floor}[1]{\left \lfloor{#1}\right \rfloor }
\definecolor{ForestGreen}{rgb}{.13,.54,.13}
\definecolor{BrickRed}{rgb}{.80,.26,.33}

\newtheorem{theorem}{Theorem}
\newtheorem{lemma}[theorem]{Lemma}

\newtheorem{proposition}{Proposition}

\newtheorem{definition}{Definition}

\newtheorem{property}{Property}

\newtheorem{example}{Example}

\renewcommand{\P}{\mathbb{P}}

\newcommand{\N}{\mathbb{N}}

\tikzset{>=latex} 
\usepackage{pgfplots} 
\usepackage{xcolor}
\usepackage[outline]{contour} 
\contourlength{1.2pt}
\usetikzlibrary{positioning,calc}
\usetikzlibrary{backgrounds}
\pgfmathdeclarefunction{gauss}{3}{%
  \pgfmathparse{1/(#3*sqrt(2*pi))*exp(-((#1-#2)^2)/(2*#3^2))}%
}
\pgfmathdeclarefunction{cdf}{3}{%
  \pgfmathparse{1/(1+exp(-0.07056*((#1-#2)/#3)^3 - 1.5976*(#1-#2)/#3))}%
}
\pgfmathdeclarefunction{fq}{3}{%
  \pgfmathparse{1/(sqrt(2*pi*#1))*exp(-(sqrt(#1)-#2/#3)^2/2)}%
}
\pgfmathdeclarefunction{fq0}{1}{%
  \pgfmathparse{1/(sqrt(2*pi*#1))*exp(-#1/2))}%
}

\colorlet{mydarkblue}{blue!30!black}

\usepgfplotslibrary{fillbetween}
\usetikzlibrary{patterns}
\pgfplotsset{compat=1.12} 

\def\N{50}
\newcount\Comments  
\definecolor{darkgreen}{rgb}{0,0.6,0}
\newcommand{\kibitz}[2]{\ifnum\Comments=1{\textcolor{#1}{{#2}}}\fi}
\newcommand{\rmr}[1]{\kibitz{red}{[RM:#1]}}
\newcommand{\gan}[1]{\kibitz{blue}{[GG:#1]}}

	    \newcommand{\red}[1]{{\leavevmode\color{red}{#1}}}

\newcommand\todo[1]{{\red{TODO: {#1}}}}

\newcommand\toref{\red{[REF]}}

\newcommand{\ppm}{p}

\newcommand{\newpar}[1]{\vspace{-1mm}\paragraph{#1}}
\def\calD{\mathcal{D}}
\def\calN{\mathcal{N}}
\newcommand{\ind}[1]{\llbracket #1 \rrbracket}
\def\ol{\overline}
\def\ul{\underline}
\usetikzlibrary{positioning,fit,shapes.arrows,shadows,backgrounds}
\newcounter{module}
\tikzset{
  modulematrix/.style={draw=blue!50!red,rounded corners,matrix of nodes,row sep=1cm,column sep=1cm,nodes={draw=green!70,align=center,font=\sffamily},inner ysep=0.2cm},
  module/.style={rounded corners, align=center, font=\sffamily, thick},
  simple module/.style={module, top color=blue!10, bottom color=blue!35, draw=blue!75, text width=60mm, minimum height=5mm},
  module down arrow/.style={module arrow, shape border rotate=-90},
  module up arrow/.style={module arrow, shape border rotate=90},
  module right arrow/.style={module arrow},
    module left arrow/.style={module arrow, shape border rotate=-180},
module arrow/.style={single arrow, single arrow head extend=2.5mm, draw=gray!75, inner color=gray!20, outer color=gray!35, thick, shape border uses incircle, anchor=tail,minimum height=0.4cm},
}

 \definecolor{DarkBlue}{rgb}{0.1,0.1,0.5}
\definecolor{DarkGreen}{rgb}{0.1,0.5,0.1}
\usepackage{verbatim}
\usepackage{microtype}
\usepackage{subfigure}
\usepackage{enumitem}
\usepackage{booktabs} 
\usepackage{color}
\usepackage{comment}
\usepackage{thmtools, thm-restate}
\usepackage{xcolor}
\makeatletter
\newcommand{\printfnsymbol}[1]{%
  \textsuperscript{\@fnsymbol{#1}}%
}
\makeatother
\usepackage{wrapfig}
\usepackage{verbatim}
\usepackage{microtype}
\usepackage{subfigure}
\usepackage{enumitem}
\usepackage{booktabs} 
\usepackage{color,graphicx}
\usepackage{comment}
\usepackage{thmtools, thm-restate}

\usepackage{amssymb,amsmath,amsthm,dsfont,bm,mathtools}

\author{Ganesh Ghalme\printfnsymbol{1}}
\author{Reshef Meir\printfnsymbol{2}}
\affil{\printfnsymbol{1}Indian Institute of Technology  Hyderabad, India}
\affil{\printfnsymbol{2}Technion-Israel Institute of Technology, Haifa, Israel}
 
\begin{document}
 \title{Condorcet's Jury Theorem with    Abstention}
 \maketitle 
\begin{abstract}
The well-known Condorcet's Jury theorem posits that  the majority rule selects the best alternative among two available options with probability one, as the population size increases to infinity. We study this result under an asymmetric two-candidate  setup, where supporters of both candidates may have different participation costs.  

When the decision to abstain is fully rational i.e., when the  vote pivotality is the probability of a tie,  the only equilibrium outcome is a trivial equilibrium where all voters except those with zero voting cost, abstain. We propose and analyze a more practical, boundedly rational model  where voters overestimate their pivotality, and show that under this model, non-trivial equilibria emerge where the winning probability of both candidates is bounded away from one. 

We show that when the pivotality estimate strongly depends  on the  margin of victory, victory is not assured to any candidate in any non-trivial equilibrium, regardless of population size and in contrast to Condorcet's assertion. Whereas, under a weak depedence on margin, Condorcet's Jury theorem is restored. 
\end{abstract}

\section{Introduction}
\gan{background: Condorcet's result.. }
Consider a population of $N$ voters voting over two alternatives  $A$ and $B$,  with $A$ being the {\em better alternative}  according to some pre-defined criterion. Consider further that the preference of each individual voter is determined independently by an outcome of a coin toss with bias $p > 0.5$ in favour of the better alternative. That is, each individual voter supports alternative $A$ with probability $p$ and $B$  with the probability $1-p$. Under this setting, the famous Condorcet's Jury Theorem states that the majority rule selects candidate $A$ with probability tending to one when population size increases to infinity. 


\gan{gap/problem: condorcet result assumes everyone votes...this is not the case in practice...presidential elections have around 60\% turnout}

An implicit assumption in Condorcet's  theorem is that \emph{everyone votes}, or at least that the decision  to vote 
  does not depend on one's preference over alternatives. 
In contrast, in many practical situations such as  political elections, or a local or national referendum, abstention is found to be a  common and prominent phenomenon. For instance, the voter turnout in United States presidential elections has been around 52\%-62\% over the past 90 years \cite{wiki:Voter_turnout}. Abstention is also observed to be a significant phenomenon  in small-scale lab experiment~\cite{Blais,Owen1984}.



\gan{ paradox of voting..why people vote...do those who abtain affect the outcome of the election ... what happens when the size of population grows to infinity }
 From a rational, economic point of view, the surprise is not that some voters abstain, but that they vote at all a.k.a. `the paradox of voting'. As Anthony Downs claimed already in 1957,  a rational voter weighs the benefit of voting (which realizes only if the voter is pivotal) against the cost. When  the  size of the electorate is large, the expected benefit derived from affecting the outcome of the election (i.e. being pivotal) is too small to induce voting  from a significant fraction of voters~\cite{Downs57} giving rise to the paradox of voting. 

\gan{paradox of voting vs condorcets result... reality is between... does there exists a boundedly rational voting model that explains large turnout and at the same time does Condorcets result holds..strategic abstention }
While Condorcet's result holds under the assumption that everyone votes,  Down's theory of rational voting suggests that only near-zero cost voters vote in large-scale elections. 
Our goal in this paper is to understand the equilibria that arise from a plausible, heuristic-based abstention model, in an attempt to reconcile the theoretical predictions with the moderate turnout rates we see in practice. The most important questions we focus on are: (1) how many equilibrium points are induced by the abstention  model and where are they located? (2) does the winning probability of the better candidate approach 1 as predicted by Condorcet's Jury theorem  in any of these equilibria?  and (3) what happens to the equilibrium points as population size increase? 

 We now present some important models for voter turnout presented in the literature and then move on to our proposed voter turnout model that depends on heuristic-based perceived pivotality of individual votes.



\label{sec:intro}
\gan{review of theoretical voting models}
  
  \newpar{\bf The Calculus of Voting model:}  
  Originally proposed by \cite{Downs57} and later developed by 
 \cite{riker68}, this  model   attributes each  voter's decision to abstain from voting to expected cost-benefit analysis. Let $p_i$ denotes the perceived \emph{pivotality} of voter~$i$, whereas $\texttt{V}_i$ denotes the  personal benefit she receives if her preferred candidate wins an election,   $\texttt{D}_i$ denotes the social benefit she receives by performing a civic duty of voting and $\texttt{G}_i$ denotes   costs of voting she incurs. These costs  include the cost of obtaining and processing information and the actual cost of  registering and going to polls (see also~\cite{Aldrich}  for discussion of voting and rational choice).  A voter $i$ votes if and only if 
 \begin{equation}
     p_i\cdot \texttt{V}_i + \texttt{D}_i \geq \texttt{G}_i. 
     \label{eq:calculus}
 \end{equation} 

 The calculus of voting  model considers  $p_i$ to be the probability that $i$'s vote would change the outcome a.k.a. probability of an event that all voters except $i$ reach a tie. The tie probabilities are derived from the aggregated stochastic votes, and thus the pivot  computation and subsequent equilibrium analysis quickly become intractable as the number of individual voters increases. 
 

\rmr{Aldrich: there is a paradox in empirical data, where studies based on aggregate data show a strong connection of $p$ and turnout, and studies based on survey data show the opposite. I'm not sure I understand the explanation, or the difference between these studies :(   It may suggest some evidence that voters mis-estimate their pivotality}

 Myerson and Weber studied voting equilibria in  a 3 or more candidate elections by fixing a distribution over preferences and considering a large population sampled from this distribution, so that candidates' scores are multinomial variables~\cite{MyersonWeber}. While tie probabilities vanish in the limit, they observe that  the \emph{relative} tie probabilities for each pair of candidates can still be compared and thus voters' best responses are well defined. 
 In a more recent work, Myerson~\cite{MYERSON2002219} suggested another relaxation by approximating candidates' scores with Poisson distributions.

 Crucially, in the above models, the tie probabilities used for voters' strategic calculations are derived directly from the vote distribution, either exactly or approximately.
    When  the size of the electorate is large, the tie probability  
 $p \approx 0$. If voting is costly, these models predict very low turnout---essentially  that only the zero-cost voters vote when the population increases to infinity.

 \newpar{\bf Heuristic pivot probabilities:} Some models assume that the voters act based on estimated or even completely wrong beliefs.  One such heuristic is derived by Osborne and Rubinstein's \emph{sampling equilibrium}~\cite{OSBORNE2003434}, where each voter estimates candidates' scores based on a small random sample of other voters.  That is, even with a large  population when voters compute their beliefs based on limited information, there is a non-negligible fraction of voters who believe they are pivotal. 

 \rmr{the pivot probability is like Binomial but with the sample size $k$ instead of $m$. If $k$ is fixed then $p\cong (1-m^2)^k$ which is sort of tie-sensitive since there is no dependency on $n$ at all. If $k$ is increasing with $n$ then it is vanishing. }
\newpar{ \bf Non-probabilistic uncertainty:} Other models stir away from  probability calculations and consider other voters' heuristics, based on dominated actions, minmax outcomes or regret~\cite{Aldrich,Ferejohn75,Meir14,Merrill1981}. 
The \emph{margin} of victory plays a major role in the voter's perceived  pivotality. Controlled experiments on voters' response to poll information show that \emph{strategic}  voting is more frequent when the winning margin is small~\cite{meir2020strategic,Ferejohn75}. 
~\cite{fairstein2019modeling} show  that voters' actions are more consistent with various heuristics based on the margin than with `rational' utility maximization models.  The key point in these models is that voters form a heuristic belief of their pivotality based on the voting profile, and that an `equilibrium' means the belief justifies itself, not that it is necessarily correct~\cite{Aldrich, 2018Meir}.  

  \rmr{here too we can write much less about these models. the key point is that voters form a heuristic belief of their pivot probability based on the voting profile, and that an equilibrium means the belief justifies itself. Another such model is by Rubinstein (see reference and details in my book)}

\medskip
Analyzing every such model separately would be tedious and leave us with an isolated set of narrow results. Instead, we identify the main factors common to all of these models, and suggest a flexible framework that captures a wide range of possible rational and boundedly-rational behaviours, without committing to one model in particular.\footnote{Although some of our theoretical results make stronger assumptions for technical reasons.} 
We still consider that the voters seek to maximize their own expected utility as in the calculus of voting model (Eq.~\eqref{eq:calculus}). However, in contrast to the calculus of voting model, in our model voters decide to abstain based on \emph{heuristically} estimated perceived pivotality. This estimate is also a function of the margin of victory as in the latter models we mentioned. Our proposed heuristic model for computation of perceived pivotality is inline both with empirical~\cite{edlin2007voting,fedderson12} and experimental~\cite{Blais,fairstein2019modeling} findings, and with the models developed therein that aim to explain high turnout in  elections. 



 We emphasize here that we do not claim to provide a new universal model of strategic voting.  Rather, we want to capture a broad class of models that share similar properties (dependency on margin and population size) in order to provide results that are not model-specific. 
\paragraph{Contributions}
We show that the set of equilibrium points  induced  by proposed abstention model consists of equilibrium points predicted by the fully rational  (such as calculus of voting) model that satisfy Condorcet's Jury Theorem (CJT) and some more \emph{non-trivial} equilibrium points.  Our results show that there are \emph{non-trivial} equilibria where the popular  candidate is more likely to win, but the winning probability is bounded away from 1. 
This probability, perhaps surprisingly, depends only on the distribution of voting costs in the population. 
The bottom line of this paper  is the result that under a plausible abstention model, induced non-trivial equilibrium points evade both; Down's paradox of voting (only a small  fraction of homogeneous voters vote) as well as CJT (the win probability of popular candidate approaches 1).

\section{Model}
\label{sec:model}
We study a two-candidate (referred to as $A$ and $B$)   election with   $N$ voters.  Each  voter  is either a supporter of $A$ (prefers candidate~$A$, i.e. $A \succ_{i} B$)  or a supporter of $B$. 
We adopt the classical calculus of voting model  as follows. We denote  as  $ c_i := \max \{ 0,\frac{\texttt{G}_i - \texttt{D}_i}{\texttt{V}_i} \} $ the \emph{effective} cost of voting for voter $i$  with $\texttt{G}_i$, $\texttt{D}_i$ and $\texttt{V}_i $ as defined in Eq.~\eqref{eq:calculus}. 

The core supporters (voters with zero effective voting cost) derive more utility from voting for their preferred candidate  than costs incurred  in participating in the voting process.  Voters with non-zero effective costs, on the other hand, vote only if the perceived pivotality of their vote exceeds the voting cost.  For mathematical convenience, we  normalize the effective cost of voting to lie between 0 and 1.  



We consider the effective cost of voting and the preference of voter~$i$ as  an independent sample from a commonly known  joint distribution $\mathcal{D}$ over $[0,1] \times \{A,B\}$ and is denoted by tuple $(c_i, T_i)$. We assume that $\calD$ has no atoms, expect possibly at $c=0$ and/or $c=1$. 

Without loss of generality, we assume  that (weakly) more voters support $A$  in expectation. 

We therefore consider $A$ as the ``better'' or ``popular'' candidate, and see the outcome where $A$ is elected as preferable.

\medskip
As discussed above, Condorcet Jury Theorem states that  \emph{if all voters vote}, the probability that $A$ wins goes to 1 as the number of voters grows.
However, this is clearly not always true if $A$ voters are more likely abstain than $B$ voters. This, in turn, may depend both on \emph{individual voters' costs}, and on \emph{voters' perceived pivotality} (which we assume to be the same for all voters). 


\paragraph{Perceived pivotality}
In the discussion above, we highlighted that the two most important factors that voters use to determine the probability of various outcome are (1) the size of voting population, $n$; and (2) the margin, $m$. Our simplifying assumption (following e.g.~\cite{MyersonWeber}) is that voters only consider $n$ and $m$ as expected values. Thus, a general \emph{pivotality model} is specified by a function $p$, which maps any pair of $n$ and $m$ to a pivot probability $p(n,m)\in[0,1]$, and is continuous and non-increasing in both parameters. We later provide several concrete pivotality models, which can either approximate the real pivot probability or reflect subjective beliefs. 

 \paragraph{Support functions}
While  $\calD$ contains all necessary information regarding the distribution of costs in the population, we would like to present costs in a more intuitive way. 

For $T\in \{A,B\}$ 
let $s_T:[0,1]\rightarrow [0,1]$  be continuous, non-decreasing function, where $s_T(c)$ should be read as the fraction of the entire population (except voter $i$) that prefers $T$ and  has individual cost at most $c$. We call $s_T$ the \emph{support function} of $T$, and note that it does not depend on the identity of voter $i$.
\begin{restatable}{proposition}{propOne}
    Any distribution $\calD$ induces a unique pair of support functions $s_A,s_B$ with $s_A(1)+s_B(1)=1$, and vice-versa.
\end{restatable}

\subsection{Issues and Elections}
An \emph{issue} is a triple $I=(s_A,s_B,p)$, i.e. it describes the support in the population for the two possible decisions $A$ and $B$ on the issue, as well as the perceived pivotality of voters on that issue.

An issue does not yet induce a valid game, since the set of players is not well defined. 
An \emph{Election} $E=(I,N)$ is an issue $I$ together with a specific number of players $N$, who are assumed to be sampled i.i.d. from distribution $\calD$ (or, equivalently, from the support functions).

An election is therefore a game, where the type of each voter/player is a random variable. Each player has two actions: vote for her preferred candidate or abstain.\footnote{Voting for the less preferred candidate is strictly dominated so we ignore this option.} A voter of type $(T_i, c_i)$ suffers a cost of $c_i$ if she votes, irrespective of the outcome, and gains 1 if her preferred candidate $T_i$ wins. 

A pure strategy profile in this game maps voter types to actions (vote/abstain).

We argue that pure Bayes-Nash equilibria in this game have a very simple form, in which all voters with cost lower than some threshold vote.




\begin{restatable}{proposition}{PropEquilibrium}(Election Equilibrium) For a given election  $(s_A,s_B,p,N)$,  every pure Bayes-Nash equilibrium can be described by a single number $c\in[0,1]$ s.t.:
\begin{align}
    n(c)&=(s_A(c)+s_B(c))N \label{eq:n_c}\\
    m(c)&=\frac{|s_A(c)-s_B(c)|}{s_A(c)+ s_B(c)}\label{eq:m_c}\\
    c&=p(n(c),m(c)), \notag
    \end{align}
    and each voter $i$ votes iff $c_i\leq c$.
\end{restatable}

For technical reasons we will assume throughout the paper that all derivatives of of $s_A,s_B$ are bounded. 

Denote by $C^*(I,N)\subseteq [0,1]$ the set of all equilibrium points of election $E=(I,N)$.

\begin{restatable}{proposition}{propTwo}\label{prop:eq_exists}
Every election has at least one equilibrium.
\end{restatable}
The proof follows from the fact that $f(c):=p(n(c),m(c))$ is a continuous function from $[0,1]$ onto itself and therefore must have
a fixed point.

\paragraph{Issue equilibrium}
An issue can be thought of as a series of elections, one for every population size $N$.

\begin{definition}[Issue equilibrium]An equilibrium of issue $I$ is a series of points $\overline{c}=(c_N)_{N}$ s.t. $\forall N, c_N\in C^*(I,N)$, and $\overline{c}$ has a limit. We denote the limit by $c^*$.
\end{definition}

For an issue equilibrium $\ol c$ with limit $c^*$, if  $c_N>c^*$ for all $N$ we say that $\overline{c}$ is a \emph{right equilibrium}. Similarly, if  $c_N<c^*$ for all $N$ we say that $\overline{c}$ is a \emph{left equilibrium}.

\paragraph{Trivial equilibrium}
An equilibrium of issue $I$ is \emph{trivial} if its limit is $c^*=0$, meaning only core supporters vote. Otherwise it is nontrivial. Clearly a trivial equilibrium is always a right equilibrium. 

\paragraph{Jury theorems} For a given election $E=(I,N)$ and a given equilibrium $c\in C^*(I,N)$, we can compute the probability of various events. One such event of interest is the victory of the popular candidate $A$.\footnote{We emphasize that we, as `outsiders' to the election, care about the \emph{actual} probability of the event, which is not affected by the perceived pivotality model $p$, once $c$ is determined.}

Formally, we define for a given election $E=(I,N)$ and equilibrium $c\in [0,1]$, a random variable  counting the number of active votes for $A$ (and likewise for $B$):
$$V_A := \sum_{i\in N}\ind{c_i\leq c \wedge T_i=A},$$
where $\{(c_i,T_i)\}_{i\in N}$ are voters sampled from $\calD$.
 Note that $V_A$ is a Binomial variable with parameters $N$ and $s_A(c)$, and expected value $s_A(c)N$.

We further denote the \emph{Winning Probability} of $A$ in equilibrium $c$ of election $E=(I,N)$ as
$$\mathbb{WP}_A(I,N,c):=\Pr(V_A > V_B| I,N,c).$$

Finally, for any issue $I$ with issue equilibrium $\ol c$, we define
$$\mathbb{WP}_A(I,\ol c):=\lim_{N\rightarrow \infty}\mathbb{WP}_A(I,N,c_N).$$

We say that the instance admits a \emph{jury theorem} at $c^*\in [0,1]$, if there is  an instance equilibrium $\ol c$ with limit $c^*$ s.t. $\mathbb{WP}_A(I,\ol c)=1$.

Similarly, $I$ admits a \emph{non-jury theorem} at $c^*$ if $\mathbb{WP}_A(I,\ol c)<1$, meaning that regardless of the size of the population, there is some constant probability that the better candidate $A$ will lose.  $I$ admits a \emph{strong non-jury theorem} if $\mathbb{WP}_A(I,\ol c)=\frac12$.

\section{ Perceived Pivotality Models}
\label{sec:tie}


We first describe models where the perceived pivotality $p(n,m)$  represents the actual tie probability $V_A=V_B$ (so a single vote would determine the  winner of the election). 
We present two  close approximations of the actual tie probability that have been used in the literature, and that can be presented as pivot functions in our model. 

\subsection{Fully Rational models}
\label{ssec:fully rational}
\rmr{We should have a proposition for the `correct' pivot probability (which can be computed either as the expected ti-probability over $n'$ sampled from $Bin(N,n/N)$; or as the probability that in the multinomial distribution $(N,(s_A,s_B,1-s_A-s_B))$ the first two elements are equal). 
This is easy, and then we need another proposition bounding this by $\Theta(p(n,m))$ of the Binomial PPM.}



The first model maintains the Binomial distribution, but assumes that the $n$ active voters are the entire population,\footnote{\label{fn:double_bin}The true pivot probability can also be written as a function of $m$ and $n$: it follows the Binomial model $p(n',m)$ from Eq.~\eqref{eq:Binom}, but where the number of active voters $n'$ is a random variable sampled from $Bin(N,n/N)$. Numerically this function is nearly identical to Eq.~\eqref{eq:Binom}, and in particular it has strong vanishing pivotality.\rmr{needs proof}}
 as in the original Condorcet Jury Theorem and in early Calculus of Voting models~\cite{riker68}. 
 \begin{example}[Binomial PPM] 
 \begin{equation}
  \ppm(n,m)=\Pr_{x\sim \text{Bin} \big ( n, (1+m)/2 \big)}(x = \floor{n/2}).\label{eq:Binom}
  \end{equation}
 \end{example}

  A later model by~\cite{Myerson1998} suggested drawing the scores of each candidate independently from a Poisson distribution. 
 \begin{example}
  [Poisson PPM]
  \begin{equation}
      \ppm(n,m)=\Pr_{\substack{x_A\sim \text{Poisson}((1+m)n/2)\\ x_B\sim \text{Poisson} ((1+m)n/2)}}(x_A=x_B)\label{eq:Poisson}
  \end{equation}
 \end{example}
Conceptually, the Poisson model is more appropriate in situations where voters can abstain (as the total number of active voters is not fixed). 
However $p(n,m)$ behaves very similarly to the Binomial model, and for our purpose they are almost the same.
In fact, both models belong in a much larger class of PPMs:
\begin{definition}[Vanishing Pivotality]
 We say that a PPM has strong [weak] vanishing pivotality if   $\lim_{n\rightarrow \infty}p (n,m)=0$ for all  [$ m > 0$] $m\geq 0$.
\end{definition} 

As we will later see, instances with weak vanishing pivotality always admit a trivial equilibrium. 
Clearly at the trivial equilibrium, Jury theorems are irrelevant: the candidate with more core support always wins with probability that approaches $1$ as the population grows, regardless of who is more popular overall. It is not hard to verify (e.g. using Stirling approximation) that in both the Binomial and Poisson PPMs, $p(n,m)=\Theta(\frac{1}{\sqrt{n}})$ for  $m=0$, and decreases exponentially fast in $n$ for any $m>0$. 

 \subsection{Tie-Sensitive models}
We saw that even in the rational model (which have strong vanishing pivotality), the case of $m=0$ is different, with a substantially higher probability to be pivotal. 

A simple and perhaps more cognitively plausible assumption is that voters \emph{consider themselves pivotal} if the margin is small enough, regardless of the number of voters.\footnote{An alternative justification that maintains full rationality is an argument that is sometimes made about large elections being `more important', as a possible explanation for the paradox of voting~\cite{Downs57}. While this suggestion fails to explain high turnout in general (unless importance increases \emph{exponentially} with $n$), even a modest increase in importance of $\Omega(\sqrt{n})$ is enough to justify Tie-sensitivity.} 

\begin{definition}[Tie-sensitive pivotality]We say that a PPM has $q$-\emph{tie-sensitive pivotality} if $p(n,0)\geq q$ for all $n$.
\end{definition}
That is, if the expected outcome is a tie, everyone thinks they are pivotal at least to some extent, regardless of the number of active voters. 

For example,~\cite{OSBORNE2003434}  suggested a model where each voter decides whether they are pivotal based on a fixed number of other voters they sample i.i.d. from the population. This induces a PPM that is identical to the Binomial one, except that the parameter $k$ is used instead of the number of voters $n$.  Thus for a fixed $k$, the `Sampling PPM' is $q_k$-tie-sensitive for some $q_k=\Theta(\frac{1}{\sqrt{k}})$. 

Another tie-sensitive PPM is induced by assuming a polynomial dependency on both $m$ and $n$. 
\begin{example}[Polynomial PPM] For $\alpha,\beta>0$
$$p(n,m)=\min\{q,\frac{1}{m^\alpha n^\beta}\}$$
\end{example}
It is immediate from the definition that the Polynomial PPM is both $q$-tie-sensitive and has a weakly vanishing pivotality. 
One particularly natural choice is when $\beta=\frac12$ (so that the dependency on $n$ is similar to the previous models), and $\alpha=1$, i.e. linear dependency on the margin.

\section{Characterizing Equilibrium Limits}
We begin by characterizing the trivial equilibrium.

\begin{restatable}
{proposition}{propThree}\label{prop:strong_trivial}
    Suppose $s_A(0)+s_B(0)>0$. Any  instance with strongly vanishing pivotality admits \textbf{only}  a trivial equilibrium.
\end{restatable}

\begin{restatable}{proposition}{propFour}\label{prop:weak_trivial}
    Suppose $s_A(0)\neq s_B(0)$. Then any instance with weakly vanishing pivotality admits a trivial equilibrium.
\end{restatable}

The proofs of Propositions \ref{prop:strong_trivial} and \ref{prop:weak_trivial} are given in Appendix. Intuitively, the proofs of both the propositions rely on the fact that we can always consider sufficiently large population size $N$, thereby making $p(n,m)$ small enough. 

Next we  show that that  the intersection points of the support functions form the limiting points of equilibrium points.  

\begin{definition}[Pivot Points]
For a given pair of support functions $s_A,s_B$, a \emph{pivot point} is any  $c\in(0,1)$ where $s_A(c)=s_B(c)>0$, and $s_A(c+\eps)\neq s_B(c+\eps)$ for all $\eps>0$. 
\end{definition}
A pivot point is a cost $c$ such that if  all voters with $c_i\leq c$ vote then both candidates have equal support, and thus the expected margin at a pivot point  is 0. 

\begin{theorem}\label{thm:nontrivial_eq}
    Let $I$ be an instance with a PPM a that has weakly vanishing  \emph{and} $q$-tie-sensitive pivotality.  Any pivot point $c^*<q$ has a right equilibrium with limit $c^*$. Moreover, for support functions that intersect at a finite number of points, the limit of any equilibrium is either 0 or some pivot point $c^*$.
\end{theorem}
\begin{proof}
We start with existence. Let $c^*<q$ be some pivot point of $I$, and let $\delta>0$. We need to show there is some $N_\delta$ and some $c_\delta\in(c^*,c^*+\delta)$ s.t. $c_\delta\in C^*(I,N_\delta)$. 

Since all derivatives are bounded, there is some open interval $(c^*,c^*+t)$ where $s_A,s_B$ differ, and w.l.o.g. $s_A(c)>s_B(c)$ for any $c\in (c^*,c^*+t)$. Let $\ul{\delta}:=\min\{t,\delta\}$ and note that by the definition of pivot point, $\eps:=m(c^*+\ul{\delta})>0$. Also, $n(c^*+\ul{\delta})=(s_A(c^*+\ul{\delta})+s_B(c^*+\ul{\delta}))N<N$. Thus
$$p(n(c^*+\ul{\delta}),m(c^*+\ul{\delta}))\leq p(N,\eps)\xrightarrow[N\rightarrow \infty]{} 0,$$
so there is some $N_\delta$ for which $p(n(c^*+\ul{\delta}),m(c^*+\ul{\delta}))<\ul{\delta}$.

On the other hand, 
$$p(n(c^*),m(c^*))=p(n(c^*),0) \geq p(N,0)\geq q > c^*.$$
Define $f(c^*+x):=p(n(c^*+x),m(c^*+x))-(c^*+x)$ then $f$ is continuous in the range $x\in[0,\ul{\delta})$ with $f(0)>0$ and $f(\ul{\delta})<0$. From intermediate value theorem there is some $x$ where $f(x)=0$ and thus $c_\delta:=c^*+x$ is an equilibrium of $(I,N_\delta)$.

\medskip
In the other direction, assume towards a contradiction that there a nontrivial equilibrium with limit $c^*$ that is not a pivot point. Note that by our assumption of finite intersection points, and due to bounded derivatives, $s_A(c)-s_B(c)>\eps>0$ in some interval $[c^*-\delta,c^*+\delta]$.
Thus in any point in this interval the pivotality goes to 0 for sufficiently large $N$, and in particular is lower than $c^*-\delta$, which means it is not an equilibrium of $(I,\ol{N})$ for any $\ol{N}\geq N$.

\medskip
Also note that if $q$ is tight (i.e. the PPM is not $q'$-tie-sensitive for any $q'>q$) then pivot points above $q$ cannot be the limit of any equilibrium either, as the pivotality in some environment of each such point $c^*>q$ is under $q$ for sufficiently large $N$.
\end{proof}

So we have a rather complete characterization of equilibria, or at least of their limit points, in every instance. A natural question is whether any of these equilibria are ``good'' in the sense of the Condorcet Jury Theorem.

\section{Jury Theorems for Pivot Points}
As observed above, a trivial equilibrium admits a jury theorem if and only if $A$ has more core supporters, but this is not very interesting. 

Our main question regards the non-trivial equilibria, whose limits are the pivot points. Clearly any equilibrium from a side where $B$ has more support admits a non-jury (as $A$ gets less than half the votes, in expectation), but we could still hope that, say, a right equilibrium at a pivot point where the derivative of $s_A$ is higher, would admit a jury theorem. 

To provide an answer, we would need to apply a concrete PPM, so the results in this section would be limited to the Polynomial PPMs, and for simplicity we also consider support functions with finite intersection. By the results of the previous section, we know that every pivot point has a right equilibrium and left equilibrium, and that (except the trivial equilibrium) there are no others.  

We first characterize the rate at which election equilibria $c_N$ approach their limit.
\begin{lemma}\label{lemma:poly_margin}Let $I$ be an instance with a Polynomial PPM and let $\ol c=(c_N)_N$ be a  non-trivial equilibrium of $I$ with limit $c^*$. Then $c_N=c^*+\Theta(N^{-\frac{\beta}{\alpha}})$.
\end{lemma}
That is, the point $c_N$ must be at distance that decreases proportionally to $N^{-\frac{\beta}{\alpha}}$: not closer neither farther away. 
\begin{proof}
We will prove for $c_N>c^*$. The proof for $c_N<c^*$ is symmetric.

    By Theorem~\ref{thm:nontrivial_eq}, $c^*$ is a pivot point. 

    Since the first and second derivatives of $s_A,s_B$ are bounded, so are those of the functions $n(c)$ and $m(c)$ (as per Eq.~\eqref{eq:n_c},\eqref{eq:m_c}). Also note that $n'(c)\geq 0$ everywhere and $m'(c^*)>0, n'(c^*)>0$ by the definition of pivot point. We also denote $s^*:=s_A(c^*)+s_B(c^*)$ and note that $s^*\cdot N = n(c^*)$.

    Therefore there must be some interval $[c^*,c^*+\delta]$ and constants $\ol m,\ul m,\ol n,\ul n>0$ such that $m'(c)\in [\ul m,\ol m]$ and $\frac{n'(c)}{N}\in[\ul n,\ol n]$ for all $c\in [c^*,c^*+\delta]$. Moreover, $\ul m,\ol m$ can be arbitrarily close to $m'(c^*)$ and likewise for $n'$ (the functions are nearly-linear near $c^*$).  We consider only $N$'s large enough so that $c_N$ is in this interval, and so that $p(n(c_N),m(c_N))<1$.

    Now, let $\eps:=c_N-c^*$, and w.l.o.g. $\eps<\min\{0.1,0.1/\ol n\}$. We want to show $\eps=\Theta(\frac{1}{N^\frac\beta\alpha})$. From the definition of equilibrium,
    \begin{align*}
        c^*+\eps &= c_N = p(n(c_N),m(c_N)) = \frac{1}{m(c_N)^\alpha n(c_N)^\beta}\\
        &=\frac{1}{m(c^*+\eps)^\alpha n(c^*+\eps)^\beta} \Rightarrow
    \end{align*}
    \begin{equation}\label{eq:m_c_star_bound}
        m(c^*+\eps) 
        = \frac{1}{(c^*+\eps)^\frac1\alpha n(c^*+\eps)^\frac\beta\alpha}.
    \end{equation}
        
    Now, by the constant bounds on the derivatives, 
    \begin{align}
        m(c^*+\eps) &\in [m(c^*)+\ul m \eps, m(c^*)+\ol m \eps] \notag\\
        &=  [\ul m \eps, \ol m \eps], \label{eq:m_u_o}\\
        n(c^*+\eps) &\in [n(c^*)+\ul n \eps\cdot N, n(c^*)+\ol n \eps\cdot N] \notag \\
        &= [(s^*+\ul n \eps)N, (s^*+\ol n\eps)N]\notag \\
        &\subset [s^* N, (s^*+0.1)N] \label{eq:s_star}
    \end{align}
For the upper bound, we get from Eqs.\eqref{eq:m_c_star_bound},\eqref{eq:m_u_o},\eqref{eq:s_star} 
\begin{align*}
    \ul m \eps &\leq m(c^*+\eps) \leq \frac1{(c^*)^\frac1\alpha (s^* N)^\frac{\beta}{\alpha}} \Rightarrow\\
    \eps & \leq \frac1{\ul m\cdot (c^*)^\frac1\alpha (s^*)\frac{\beta}{\alpha} N^\frac{\beta}{\alpha}}=O(\frac{1}{N^\frac{\beta}{\alpha}}),
\end{align*}
since $\ul m, c^*$ and $s^*$ are all constants. 
Likewise, for the lower bound, 
\begin{align*}
     \ol m \eps &\geq m(c^*+\eps) \geq \frac1{(c^*+0.1)^\frac1\alpha ((s^*+0.1) N)^\frac{\beta}{\alpha}} \Rightarrow\\
     \eps &= \Omega(\frac{1}{N^\frac{\beta}{\alpha}}) \text{ as required.}
\end{align*}
\end{proof}
It is also interesting that the exact rate of convergence depends on the derivative of $m(c)$, but not on that of $n(c)$, meaning that the individual slopes of $s_A,s_B$ at $c^*$ don't matter, only the difference between them.

\subsection{Computing the win probability}
We denote by $\Phi$ the CDF of the Normal distribution, i.e. $\Phi(x)=\Pr_{X\sim \calN(0,1)}(X<x)$.
\begin{proposition}[Characterization of Jury Equilibria]
    Let $I$ be an instance with a Polynomial PPM with parameters $q,\alpha,\beta>0$, and let $\ol c$ be a nontrivial right- or left-equilibrium of $I$ with limit $c^*<q$ and popular candidate $A$. Then there are three cases:
   \begin{itemize}
        \item If $\alpha>2\beta$,  then $\mathbb{WP}_A(I,\ol c)=1$;
        \item If $\alpha<2\beta$, then $\mathbb{WP}_A(I,\ol c)=\mathbb{WP}_B(I,\ol c)=0.5$;
        \item If $\alpha=2\beta$, then $\mathbb{WP}_A(I,\ol c)=\Phi\left((c^*)^\frac{-1}{\alpha}\right)$.
    \end{itemize}
\end{proposition}
That is, we get a phase transition at $\alpha=2\beta$, where above the threshold all non-trivial equilibria admit a jury theorem, below the threshold they admit a strong non-jury theorem, where both candidates are equally likely to win, and at the threshold (including that natural case of $\alpha =1, \beta=0.5$ mentioned above), we still have a non-jury theorem where the popular candidate only has some constant $<1$ win probability. Interestingly, this probability does not depend at all on the shape of the support functions, only on the position of the pivot point. 
\begin{proof}
    Let $\ol c=(c_N)_N$ be a nontrivial equilibrium with limit at pivot point $c^*<q$. 
    The random variable $V_A$ is a Binomial variable with parameters $p_A=0.5(1+m(c_N))$ and  $n=n(c_N)$, and $A$ wins if $V_A>n/2$.\footnote{To be completely accurate, we should sample $v_A$ from a Binomial distribution with a number of samples that is also a Binomial variable (see Footnote~\ref{fn:double_bin}), but the difference is negligible.}
    
    Since $n(c_N)=\Omega(N)$, for large enough $N$ we can approximate the distribution of $V_A-n/2$ with a Normal distributions with mean and variance as follows:
    \begin{align*}
        \mu_A &= 0.5m(c_N)n(c_N); \\
        \sigma_A^2 &= 0.5(1+m(c_N))0.5(1-m(c_N))n(c_N) \\
        &= 0.25(1-(m(c_N)^2)n(c_N)= 0.25(1-(m(c_N)^2)n(c_N);
    \end{align*}


    From Lemma~\ref{lemma:poly_margin}, there are some constants $\ul t,\ol t>0$ s.t. 
     
     $$\ul t\cdot N^{-\frac{\beta}{\alpha}}< m(c_N)<\ol t\cdot N^{-\frac{\beta}{\alpha}},$$
     and we know from Eq.~\eqref{eq:n_c} that
     $s^*N < n(c_N) < (s^*+0.1)N.$
     \begin{align*}
          \mu &>0.5 \ul t\cdot N^{-\frac{\beta}{\alpha}}s^*N =  0.5 s^*\ul t\cdot N^{1-\frac{\beta}{\alpha}}, \text{and} \\
        \mu &< 0.5 \ol t\cdot N^{-\frac{\beta}{\alpha}}(s^*+0.1)N =  0.5 (s^*+0.1)\ol t\cdot N^{1-\frac{\beta}{\alpha}}.
     \end{align*}
     We can similarly bound $\sigma^2_A$, as $m(c_N)^2$ is negligible, 
     $$0.24 s^* N < \sigma_A^2 < 0.25 (s^*+0.1)N.$$
    Using the upper and lower bounds on $\mu$ and $\sigma^2$, we get the following:
    \begin{align*}
       \frac{\mu}{\sigma}&< \frac{0.5 \ol t\cdot N^{1-\frac{\beta}{\alpha}}(s^*+0.1)}{\sqrt{0.24 s^* N}} < \ol r\cdot N^{0.5-\frac{\beta}{\alpha}} \\
       \frac{\mu}{\sigma}&> \frac{0.5\ul t\cdot N^{1-\frac{\beta}{\alpha}}s^*}{\sqrt{0.25( s^*+0.1) N}} > \ul r\cdot N^{0.5-\frac{\beta}{\alpha}},
    \end{align*}
    for some positive constants $\ul r, \ol r$.

    Finally, we can use the last bounds to estimate the probability of the popular candidate winning:
    \begin{align*}
        &\mathbb{WP}_A(I,N,c_N)\cong \Pr_{x\sim \calN(\mu,\sigma^2)}(x>0) = \Pr_{x\sim \calN(0,1)}(x<\frac{\mu}{\sigma}) \Rightarrow \\
       &\Phi(\ul r N^{0.5-\frac\beta\alpha}) <  \mathbb{WP}_A(I,N,c_N)
       < \Phi(\ol r N^{0.5-\frac\beta\alpha}),
    \end{align*}
    where $\Phi$ is the CDF of the Normal standard distribution function.\footnote{As the approximation becomes more accurate with $N$, the inequalities would hold without approximation for sufficiently high $N$ and the right constants $\ul r,\ol r$.}
    We can now easily distinguish three cases, with a sharp phase transition:
    \begin{description}
        \item[ 1. $\alpha>2\beta$.] In this case, where the sensitivity to the margin is much higher than the sensitivity to the number of voters, we have that $N^{\frac12-\frac\beta\alpha}\rightarrow \infty$, so 
        $$\mathbb{WP}_A(I,\ol c)=\lim_N\mathbb{WP}_A(I,N,c_N) = 1.$$
        That is, we get a jury theorem as the more popular candidate wins w.p. 1 at the limit. 
        \item[2. $\alpha<2\beta$.] When the sensitivity to the margin is smaller than the threshold,  $N^{\frac12-\frac\beta\alpha}\rightarrow 0$, so both candidates win with equal probability at the limit. 
        \item[ 3. $\alpha=2\beta$.] This includes the $\alpha=1,\beta=\frac12$ case mentioned in the preliminaries. Here $N^{\frac12-\frac\beta\alpha}=N^0=1$, and thus $A$ wins with some probability that is bounded between two constants $[\Phi(\ul r), \Phi(\ol r)]$.
    \end{description}
    It is not hard to see that we can tighten the gap between $\ul r,\ol r$, so in the limit they both equal $r=(s^*)^{0.5-\frac\beta\alpha}\cdot(c^*)^{-\frac1\alpha}$; meaning that in the threshold case, $r=(c^*)^{-\frac1\alpha}$. In particular, this means that the winning probability is completely independent of the shape of the support functions, and only depends on the position of the pivot point $c^*$. Higher $c^*$ means a larger number of active voters and higher winning probability of the popular candidate, in the limit. 
    \end{proof}

    \rmr{For the example with the linear supports and $\alpha=1$ we indeed get $\mathbb{WP}_A = \Phi(1/0.6) = 0.9522$}

 \section{Stability of  Equilibrium Points}
 \label{sec:stability}
 In this section we show that the right side equilibrium i.e. $c^+$ is stable by  showing that when the perceived pivotality is different across voter types, a constant fraction of voters have an incentive  to participate (or abstain) such that the equilibrium $c^+$ is restored. Note that $c^0$ is clearly a stable  equilibrium since the margin is a near-constant. 
\begin{restatable}{proposition}{ProbOne}
Let $c_A \geq \widehat c$ and $c_B \geq \widehat c$ be cost thresholds for voters of type $A$ and $B$ respectively, and perceived pivotality $p$ is given as $p = p(c_A, c_B) = p(m(c_A, c_B), n(c_A, c_B))$. Then for any voting instance $I$ with $N$ voters, the best response of voters  is such that  $c_A = c^+$ and $c_B = c^+$.   
\end{restatable}


\begin{figure}
\centering
\begin{tikzpicture}
\begin{axis}[
    xlabel={$c_A$},
    ylabel={$c_B$},
    xmin=0, xmax=100,
    ymin=0, ymax=100,
    legend pos=north west,
     axis y line=left,
      axis x line=bottom,
     ticks=none
]

      \draw[fill=blue!20!white, draw=none] plot[smooth,samples=100,domain=0:0] (0,100)  -- plot[smooth,samples=100,domain=0:50] (50,100) -- plot[smooth,samples=100,domain=50:50] (50,50) -- plot[smooth,samples=100,domain=50:30] (30,0) -- plot[smooth,samples=100,domain=30:0] (0,0);

            \draw[fill=blue!20!white, draw=none] plot[smooth,samples=100,domain=50:50] (\x,\x)  -- 
            plot[smooth,samples=100,domain=50:70] (70,100)  -- 
            plot[smooth,samples=100,domain=70:100] (100,100)  -- 
            plot[smooth,samples=100,domain=100:100] (100,0) -- plot[smooth,samples=100,domain=100:50] (50,0) ;

    \node at (50,45) {$c^+$};
    
    \node at (50,50)[scale = 2.5] {$\star$};
    
 
    \node at (20, 20) {$ \circ$};
       \node at (20,15) {$\widehat c$};
    \draw[dotted] (0.07,0.25) -- (0.7,2);
   \addplot[color=red] coordinates {
    (0,0)(1,1)
    };
    \addplot[color=blue, dashed] coordinates {
    (0.3,0)(0.7,1)
    };

  \node at (30,50) [module right arrow, color=red, scale=0.65] {$  \  \ A \   $};
  \node at (70,50) [module left arrow, color=red, scale=0.65] {$\ A \ \ $};

 \node at (39,18) [rotate= 65] {$p= p^+$};

 \node at (30,18) [rotate= 65] {$p > p^+$};

 \node at (45,18) [rotate= 65] {$p < p^+$};
     \end{axis}
\end{tikzpicture}
    \caption{A schematic illustation of stability of the equilibrium point $c^+$. When $c' > c^+$ is an equilibrium estimate, the $A$ supporters from  right shaded region are incentivized to abstain whereas under $c"<c^+$ the $A$ supporters from left shaded region are incetivized to participate.}

    \vspace{15pt }
    
    \label{fig:stabilityOne} 
\end{figure}
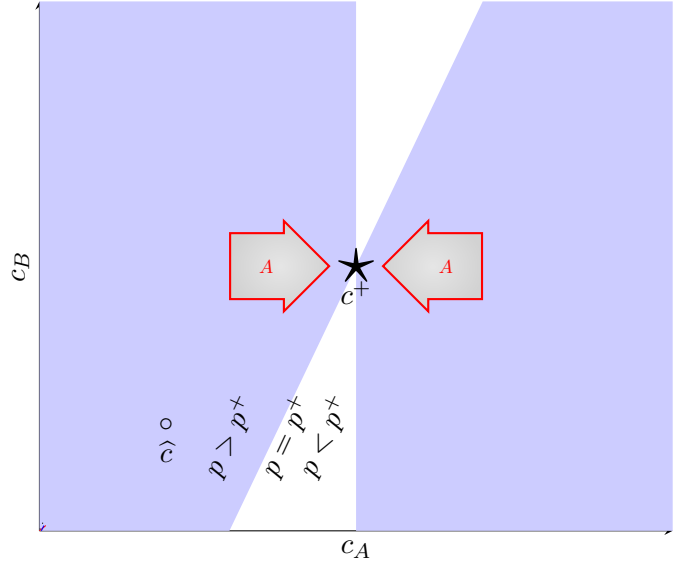

 We provide the intuition of the proof.  The detailed proof is given in Appendix. Notice in equilibrium  both types have the same cost  threshold and hence the equilibrium points must be on the line  $c_A= c_B$ as shown in Figure~\ref{fig:stabilityOne}. Suppose first that  $c_A > c_B$. In this case,  the   supporters of $A$ shown  in the shaded area (in Figure \ref{fig:stabilityOne}) to the right of $c^+$ are less pivotal and hence are   incentivized to abstain.  Similarly,  when $c_A <  c_B$,  the $A$ voters in the shaded region on the left side of $c^+$ are more pivotal and have an incentive to participate. The $B$ supporters on the other hand will participate in the first case and abstain  in the second case.  Thus, the population of participating voters adjusts itself (i.e. $c_A \downarrow c^+$ and $s_B \uparrow c^+$ when $c_A > c_B$ and $c_A \uparrow c^+$ and $s_B \downarrow c^+$ when $c_A<c_B$) such that the equilibrium $c^+$ is restored.

The left-side equilibrium, on the other, hand may not always be stable. Consider $c" < c^-$ be an equilibrium with cost thresholds $c_A$ and $c_B$ with $c_A <  c_B$. Then,  the supporters of $B$ are  incentivized to participate as, with additional participation from supporters of $B$, the margin $m$ increases, increasing  the probability of win for $B$. The supporters of $A$, on the other hand  are also incentivized to participate. An additional participation from $A$ supporters would mean that $m$ decreases. Hence stability of $c^-$ depends on the relative increase in the participation from each type of agent. 


The emergence of an unstable equilibrium between two stable ones (one of which is trivial) also occurs e.g. in markets with positive externalities~\cite{katz1985network}. In our case externalities behave non-monotonically (positive under $\hat c$ and negative above but the results are similar.

\section{Inducing Unbiased Participation} \label{sec:MD}

\rmr{added this:}
In the previous section we saw that allowing everyone to participate may introduce multiple equilibria, including some where the less-popular candidate wins with higher probability, and where the probability of the popular candidate to win is bounded regardless. 

We next explore a different idea: instead of allowing everyone to vote, we sample a sufficiently small group of voters (odd $n$), such that they are all guaranteed to perceive themselves as pivotal, and hence all of them vote. The result is unbiased voting, but with fewer voters.  A small sample size  also guarantees unique, stable equilibrium that favours the popular candidate. 



 Denote by $\varepsilon:=s_A(1)-s_B(1)$, the expected margin under full  participation, so 
 $s_A(1) =  ( 1+ \varepsilon)/2  $ then $n \leq 1/\varepsilon^2$. We then need to select sufficiently small $n$ so that $p(n,\eps)\geq 1$. In the fully rational models, this can only occur for a single active voter. However under our tie-sensitive model we  have 
 \begin{align}
1 \leq  \ppm (n, \varepsilon) & = \frac{1}{\varepsilon \sqrt{n}} 
 \end{align}
 so any $n \leq  1/\varepsilon^2$ will guarantee participation of the entire sample. 
 The perceived margin   provides sufficient information to bound the winning probability of popular candidate; independent of the shapes of support functions. 
 
 We now first consider a single round.
This already provides us with a uniform lower bound of 0.84134 on the win probability of A when the margin $\varepsilon$ goes to 0. \rmr{removed due to even $n$: and is only higher for larger margin.}  The number 0.84134 is not coincidental but is the probability that a standard Normal random variable does not exceed a single standard deviation (See figure \ref{fig:five}). Of course, the  number is  sensitive to the parameters of the exact PPM we use (as are the results in the previous sections), but not to the  shapes of the support functions. 

For larger margin, the win probability of A may be either higher or lower, and this depends mainly on the parity of $1/\eps^2$ (where for odd numbers it is always higher). 

Still, we would like to further improve the winning probability. One direction (which we do not explore in this work) is to sample a higher number of voters in an attempt to balance bias and quantity. Instead, we will explore the possibility of allowing multiple rounds.



\begin{figure}
    \centering
\begin{tikzpicture}
  \message{Cumulative probability^^J}
  
  \def\B{11};
  \def\Bs{3.0};
  \def\xmax{\B+3.2*\Bs};
  \def\ymin{{-0.1*gauss(\B,\B,\Bs)}};
  \def\h{0.07*gauss(\B,\B,\Bs)};
  \def\a{\B-0.8*\Bs};
  
  \begin{axis}[every axis plot post/.append style={
               mark=none,domain={-0.05*(\xmax)}:{1.08*\xmax},samples=\N,smooth},
               xmin={-0.1*(\xmax)}, xmax=\xmax,
               ymin=\ymin, ymax={1.1*gauss(\B,\B,\Bs)},
               axis lines=middle,
               axis line style=thick,
               enlargelimits=upper, 
               ticks=none,
               xlabel=$\# \ A \ votes $,
                ylabel=$\mathbb{P}(vote \ A) $,
               every axis x label/.style={at={(current axis.right of origin)},anchor=north},
               every axis y label/.style={at={(0,0.8)},anchor=north, rotate=90},
               width=0.7*\textwidth, height=0.55*\textwidth,
               y=700pt,
               clip=false
              ]
    
    \addplot[blue,thick,name path=B] {gauss(x,\B,\Bs)};
    
    \path[name path=xaxis]
      (0,0) -- (\pgfkeysvalueof{/pgfplots/xmax},0);
    \addplot[blue!25] fill between[of=xaxis and B, soft clip={domain=-1:{\a + 3.5}}];
    

          \addplot[mydarkblue,dashed]
      coordinates {({\a +2.2},0.15) ({\a + 2.2},{-\h + 0.01 })}
      node[mydarkblue,below=5pt] {\tiny $\mu$};


    \addplot[mydarkblue,dashed]
      coordinates {({\a -0.25},0.15) ({\a -0.25},{-\h + 0.01})}
      node[mydarkblue,below= 1pt] { \tiny $\frac{n}{2} $};

  
  \path[name path=xaxis]
      (0,0) -- (\pgfkeysvalueof{/pgfplots/xmax},0);
    \addplot[red!25] fill between[of=xaxis and B, soft clip={domain={\a +0.7}: {\a + 12}}];

      \path[name path=xaxis]
      (0,0) -- (\pgfkeysvalueof{/pgfplots/xmax},0);
    \addplot[orange!25] fill between[of=xaxis and B, soft clip={domain={\a - 1.3}: {\a + 0.7}}];


    \node[black,above left] at ({1+ 0.6*(\a)},0.02) {
    $ \tiny \mathbb{P}(B \ wins)$};
        \node[black,above right] at ({1+ 1.7*(\a)},0.02) {
    $\tiny \mathbb{P}(A \ wins)$};
        \node[black,above right] at ({1 -  0.1*(\a)},0.07) {
    $\tiny \mathbb{P}(no \ result)$};
        \node[red,above right] at ( {\B - 2.1}, 0.08) {\tiny{$ \delta $}};
        
  
    \addplot[<->, black,line width=0.3pt]
      coordinates {(\B - 2.5 , 0.08) (\B - 0.2, 0.08)};
       \addplot[->, black,line width=0.3pt]
      coordinates {(3.5 * sin(40), 0.022) (3.5, 0.004 )};
             \addplot[->, black,line width=0.3pt]
      coordinates {(22 * sin(120), 0.022) (17, 0.005 )};
       \addplot[->, black,line width=0.3pt]
      coordinates {(2 * sin(60), 0.07) (8, 0.05 )};

  \end{axis}
\end{tikzpicture}

    \caption{The figure illustrates  win probability of the popular candidate in any given intermediate round.  For majority,    $\Pr(A \ wins) = 0.84134$ and similarly for supermajority  with margin $ m \times 0.3$ (shown by orange color) it is $0.7580$ (the red area). \rmr{in the figure the orange area looks more like half the margin} }
    \vspace{15 pt}
    \label{fig:five}
\end{figure}
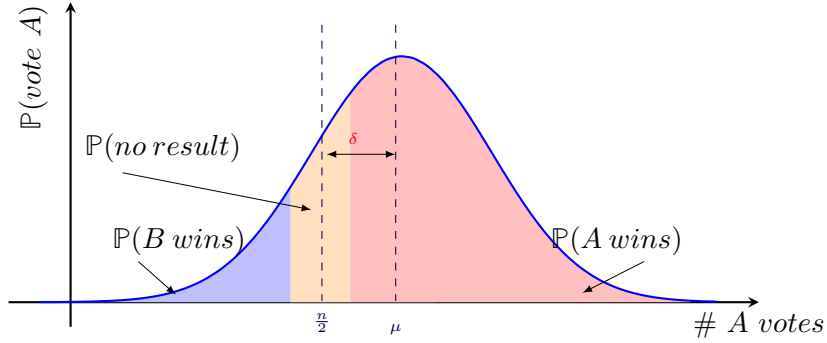

\paragraph{Multi-round voting}

Note that if we require a supermajority instead of a majority, the winning probability of A decreases, but not as much as the winning probability of B (see Fig \ref{fig:five}). So in the extreme case, we could just try again and again until we get a unanimous vote.  However such an extreme approach may substantially alter voters' behavior, which do not even know the number of rounds. 

We can thus limit voting to at most two rounds. Given an expected margin $\eps$:
\begin{itemize}
    \item Set a supermajority threshold $\eps'  < \eps$;
    \item sample $n=n(\eps)$ voters;
    \item If the first round ends with a margin at least $\eps'$, the winner is determined and there is no second round;
    \item Otherwise, we run a second round with $n'=n(\eps')$ voters and simple majority.
\end{itemize}
Note that the chances of the unpopular candidate B winning in either round are small, as the first round require a supermajority, and the second round, if occurs, has $n'>n$ voters.

\section{Numerical Study with an Example}
\label{sec:experiments}
In this section, we extend the theoretical results of the previous section with the help of an example for different model parameters, numerically. 

 \begin{restatable}{example}{exOne}
Suppose $s_A(c)=0.1+c/2$ and $s_B(c)=0.4$. That is, $B$ has 40\% overall support, all of them core supporters and $A$ has 60\% overall support and a fraction of  voters are distributed  over cost range $[0,1]$. 
\label{ex:first}
\end{restatable}
 Figure \ref{fig:fourth} shows the win probability of candidate $A$ under $c^+$ and  for a fixed population size of $10^6$ for different values of $\alpha$. The larger value of $\alpha$ pushes the equilibrium   point $c^+$ towards the right, increasing the win probability of the popular candidate. However, as shown in the blowup box, this probability is still bounded away from 1. 


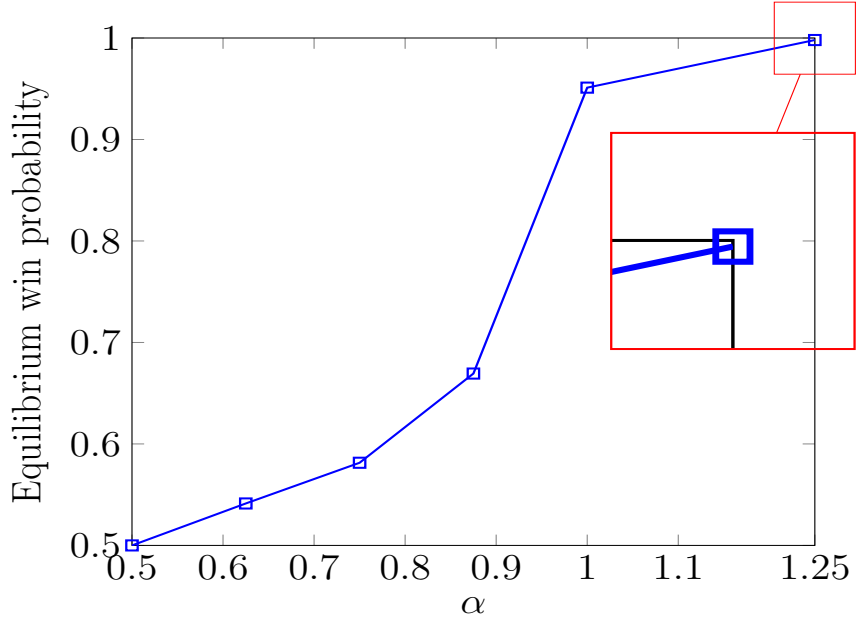
\begin{figure}[ht!]
\centering 
 \resizebox{0.7\columnwidth}{0.5 \columnwidth}{
\begin{tikzpicture}
    \begin{scope}[
        spy using outlines={
            rectangle,
            magnification=3,
            connect spies,
            size=3cm,
            blue,
        },
    ]
        \begin{axis}[
            ylabel={Equilibrium win probability},
            xlabel={$\alpha$},
            xmin=0.5, xmax=1.25,
             ymin=0.5, ymax=1,
             label style={font=\Large},
                    tick label style={font=\Large}, 
            xtick={0.5,0.6,0.7,0.8,0.9,1, 1.1, 1.25},
        ]


\addplot[ thick,
    color=blue,
    mark=square
    ]
    coordinates {
    (0.5,0.5)(0.625,0.541341) (0.75,0.58145) (0.875,0.66935) (1,0.95108) (1.25,0.998) 
    };
            \coordinate (point) at (axis cs:1.25,9.9983e-1);


                \coordinate (spy point) at (axis cs:1.16,0.8);
            \spy[color=red] on (point) in node (spy) at (spy point);
        \end{axis}
    \end{scope}


\end{tikzpicture}}
\vspace{-3mm}
    \caption{Win probability of candidate $A$ under different values of $\alpha$ in $p(n,m) = \min( 1,  \frac{1}{m^{\alpha}\sqrt{N}})$. The blowup shows that win probability is still below 1 for $\alpha=1.25$.\vspace{-3mm} }
    \label{fig:fourth}
    \end{figure}
    
\rmr{Suppose in the example of Fig 1, we fix $c^+=0.6+\eps$. How large $N$ needs to be so that $c^+$ is an equilibrium, for $\alpha=0.75, \alpha=1$ and $\alpha=1.25$? conjecture: for $\alpha<1$ we have $N=O(f(1/\eps))$ and for $\alpha>1$ we have $N=\Omega(g(1/\eps))$ where $g\gg f$ e.g. exponential vs. polynomial.  phase transition at $\alpha=1$.}

In Figure \ref{fig:win_per_size}, we observe that the  win probabilities under $c^+$ and $c^-$ for respective winning candidates are significantly different. While $B$ wins with a high probability under $c^-$, $A$ wins under $c^+$ with a high probability for the same population size. Hence  $B$ would always prefer a smaller fraction of the   population to vote whereas the popular candidate prefers a large  population to vote. Furthermore, for  a fixed population size, the win probabilities of these candidates under their favoured equilibria are not the same.

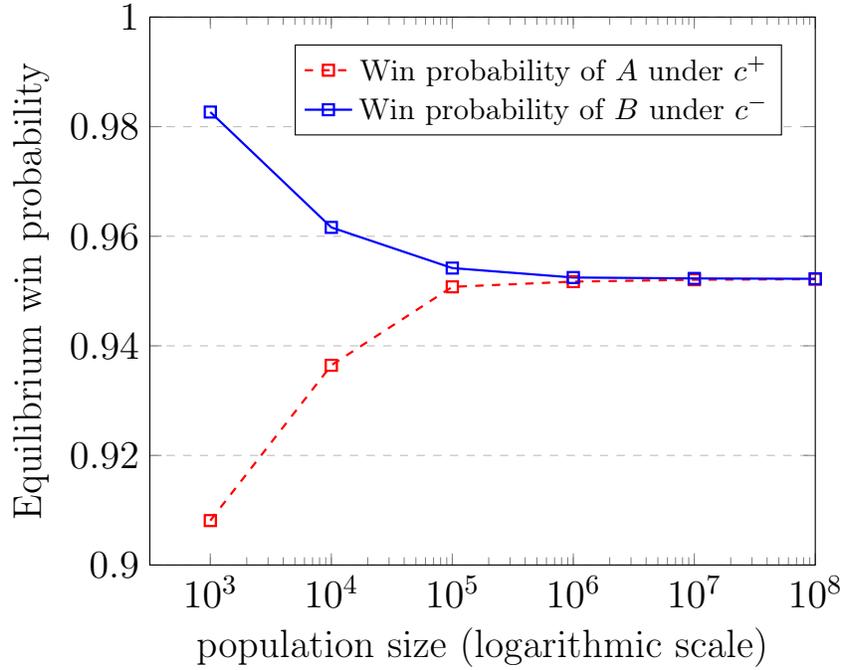
\begin{figure}[ht!]
\centering 
 \resizebox{0.7\columnwidth}{0.55 \columnwidth}{
\begin{tikzpicture}
\begin{axis}[
 xmode=log,
    log ticks with fixed point,
    xlabel={population size (logarithmic scale)},
    ylabel={Equilibrium win probability},
    xmin=0, xmax=10^8,
    ymin=0.9, ymax=1,
    label style={font=\Large},
                    tick label style={font=\Large}, 
    xtick={10^3, 10^4, 10^5, 10^6,10^7,10^8},
    xticklabels = {$10^3$,$10^4$, $10^5$, $10^{6}$,$10^7$,$10^8$},
    ytick={0.88,0.9,0.92,0.94,0.96,0.98,1},
    legend pos=north west,
    ymajorgrids=true,
    grid style=dashed,
legend style={
at={(0,0)},
anchor=north east,at={(axis description cs:0.95,0.95)}}]
\addplot[ thick, 
    color=red,
    mark=square, mark options={solid},  dashed
    ]
    coordinates {
  
   (10^3, 0.90810929)(10^4, 0.93644931)(10^5, 0.95076779)(10^6, 0.95171826) (10^7, 0.95206293)(10^8, 0.95216058)
  
    };
\addplot[ thick,
    color=blue,
    mark options={solid},
    mark=square
    ]
    coordinates {
    (10^3, 0.98266767)(10^4, 0.96161482)(10^5, 0.95419222)(10^6,  0.95247038)(10^7,  0.95230315)(10^8,0.95223574) 
    };
\legend{Win probability of $A$ under $c^{+}$,Win probability of $B$ under $c^-$}
\end{axis}
\end{tikzpicture}
}
\caption{ Win probability for different values of $N$ under respective induced equilibria.  The equilibrium win probability for a popular candidate $A$ increases with $N$ whereas it decreases for the unpopular candidate $B$. \rmr{maybe add in appendix a similar figure  for $n^*$ as a function of N}}
\label{fig:win_per_size}
\end{figure}

    \usetikzlibrary{patterns}
      \pgfdeclarepatternformonly{south west lines}{\pgfqpoint{-0pt}{-0pt}}{\pgfqpoint{3pt}{3pt}}{\pgfqpoint{3pt}{3pt}}{
            \pgfsetlinewidth{0.4pt}
            \pgfpathmoveto{\pgfqpoint{0pt}{0pt}}
            \pgfpathlineto{\pgfqpoint{3pt}{3pt}}
            \pgfpathmoveto{\pgfqpoint{2.8pt}{-0.2pt}}
            \pgfpathlineto{\pgfqpoint{3.2pt}{.2pt}}
            \pgfpathmoveto{\pgfqpoint{-.2pt}{2.8pt}}
            \pgfpathlineto{\pgfqpoint{.2pt}{3.2pt}}
            \pgfusepath{stroke}}
    
    \pgfdeclarepatternformonly{south east lines}{\pgfqpoint{-0pt}{-0pt}}{\pgfqpoint{3pt}{3pt}}{\pgfqpoint{3pt}{3pt}}{
        \pgfsetlinewidth{0.4pt}
        \pgfpathmoveto{\pgfqpoint{0pt}{3pt}}
        \pgfpathlineto{\pgfqpoint{3pt}{0pt}}
        \pgfpathmoveto{\pgfqpoint{.2pt}{-.2pt}}
        \pgfpathlineto{\pgfqpoint{-.2pt}{.2pt}}
        \pgfpathmoveto{\pgfqpoint{3.2pt}{2.8pt}}
        \pgfpathlineto{\pgfqpoint{2.8pt}{3.2pt}}
        \pgfusepath{stroke}}

 We begin by showing that for large values of $N$, non-trivial equilibria emerge   around the pivot point and converge to the pivot point as $N$ goes to infinity. 
 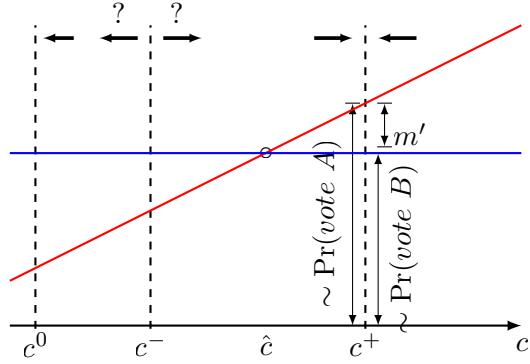
\begin{figure}[t!]
 \centering 
 \begin{tikzpicture}[scale=0.85, blend group=soft light]
\def\a{40}
\def\bx{-20}
\def\by{-13}
\draw[blue,thick] (0.5*\a+\bx, 0.4*\a+\by) -- (0.7*\a+\bx, 0.4*\a+\by);
\draw[red,thick] (0.5*\a+\bx, 0.35*\a+\by) -- (0.7*\a+\bx, 0.45*\a+\by);
\draw[thick,->] (0,0.3) -- (8,0.3);
\draw[dashed,thick] (0.639*\a+\bx,0.25) -- (0.639*\a+\bx,5);
\draw[dashed,thick] (0.555*\a+\bx,0.25) -- (0.555*\a+\bx,5);
\draw[dashed,thick] (0.51*\a+\bx,0.25) -- (0.51*\a+\bx,5);
\draw[double,->] (0.55*\a+\bx,4.8) -- (0.535*\a+\bx,4.8);
\node at (0.55*\a+\bx + 0.6,4.8 +0.4 ) {$?$};
\node at (0.55*\a+\bx - 0.31,4.8 +0.4 ) {$?$};
\draw[double,->] (0.56*\a+\bx,4.8) -- (0.575*\a+\bx,4.8);
\draw[double,->] (0.619*\a+\bx,4.8) -- (0.634*\a+\bx,4.8);
\draw[double,->] (0.659*\a+\bx,4.8) -- (0.644*\a+\bx,4.8);
\draw[double,->] (0.525*\a+\bx,4.8) -- (0.513*\a+\bx,4.8);

\node at (0.6*\a+\bx, 0.4*\a+\by) {$\circ$};

\node at (0.639*\a+\bx,-0) {$c^+$};
\node at (8,0) {$c$};
\node at (0.555*\a+\bx,-0) {$c^-$};
\node at (0.51*\a+\bx,-0) {$c^0$};

\node at (0.6*\a+\bx,-0) {$\hat c$};
\draw[|<->|] (0.639*\a+\bx+0.3, 0.4*\a+\by+0.09) -- (0.639*\a+\bx+0.3, 0.414*\a+\by + 0.23 );
\draw[<->|] (0.639*\a+\bx+0.2, 0.3) -- (0.639*\a+\bx+0.2, 0.4*\a+\by + 0.01 );
\draw[<->|] (0.639*\a+\bx-0.2, 0.3) -- (0.639*\a+\bx-0.2, 0.414*\a+\by + 0.23 );
\node[rotate=90] at (0.639*\a+\bx+0.6,1.5) {$ \sim \Pr({vote}~B)$};
\node[rotate=90] at (0.639*\a+\bx-0.6,1.9) {$ \sim \Pr({vote}~A)$};
\node at (0.639*\a+\bx+0.7, 0.4*\a+\by+0.3) {$m'$};




\end{tikzpicture}
\caption{The pivot point $\hat c$ is marked by a circle at the intersection of the support functions, with the two non-trivial equilibria on its sides (dashed lines). For the upper equilibrium $c^+$, the probability of a random voter to vote $A$ or $B$ is proportional to $s_A(c^+)$ and $s_B(c^+)$, respectively. The $m'$ is   proportional to the margin of victory. The bold arrows above indicate that $c^+, c^0$ are stable  equilibria whereas $c^-$ is often not stable.}
\vspace{15pt}

\label{fig:linearSupport}
\end{figure}

 \section{Conclusion}
   Our  results show  that under the boundedly rational PPM    the equilibrium outcome does not guarantee a  decisive win for any of the candidates even with an arbitrarily large electorate size. Interestingly, the proposed model also captures the equilibria induced in fully rational models (Binomial and Poisson models) as trivial equilibria.     Our boundedly rational model satisfies the property that    both candidates enjoy almost equal support  due to disproportionate abstention from their supporters.  The \emph{almost} equal support for candidates is a result of a trade-off between the following two factors;  
we have the  trade-off between following two factors. 
then contested by  concluding that agents also receive utility from the performing the civic duty .
a)  The candidate having majority support among high-cost voters  faces  more abstention. This increases the winning probability  of the unpopular candidate against the popular candidate and       b) voters are incentivized to vote  depending on how many \emph{other} voters  vote. Hence, when a large fraction (overall) of  voters abstain,  the size of the voting population shrinks increasing the pivotality parameter which causes more voters to vote.  

  The number of votes on the ballot increases with the voting population, however, this does not benefit  both the candidates equally; disproportionate  increase in the support of unpopular candidates leads to anti-Jury results. It is interesting to see how and if  this happens in practice. We also leave the sensitivity analysis and robustness of these results as future work. Finally, coming up with a fully rational model that explains both high turnouts and surprises in large elections is an interesting future problem.

\appendix 
\newpage 

\section{Missing Proofs}


\propOne*
\begin{proof} Suppose there are $N$ voters sampled from $\calD$.
    Given a cost threshold $c$,  let the random variable $X_A^{(i)}(c) := \frac{1}{N-1} \sum_{j \neq i}   \mathbb{1}[ c_j \leq c \text{ and } A \succ_{j} B ]$ denote fraction of supporters---other than voter $i$---of candidate $A$ having the voting cost at-most $c$ and similarly, define  $X_B^{(i)}(c):= \frac{1}{N-1} \sum_{j \neq i}^{N-1} \mathbb{1}[ c_j \leq c \text{ and } B \succ_{j} A ]$. Note that, as $X_{T}^{(i)}(c)$ does not depend on her type $T_i$ and her private cost $c_i$ and since every agent  is exposed to the same knowledge, we have, $ X_{T}(c) := X_{T}^{(1)}(c) = X_{T}^{(2)}(c) = \cdots = X_{T}^{(N)}(c)$ for all $T \in \{A,B\}$. Finally, let $ s_T(c) := \mathbb{E}[X_{T}(c)] $, and note that it depends neither on $i$ nor on $N$.

    To see that $s_T(c)$ is continuous, let $s^+, s^-$ the right and left limits of $s_T(c)$ at $c^*\in(0,1)$. If $s^-<s^+$ then $Pr_{\calD}[(c^*,T)] = s^+ - s^- >0$, in contradiction to out assumptions that $\calD$ has no atoms at $(0,1)$.

    In the other direction, we define $Pr_\calD[T_i=A]:=s_A(1)$, then by assumption $Pr_\calD[T_i=B]:=1-s_A(1)=s_B(1)$. For each $T\in \{A,B\}$, then function $s_T(c)/s_T(1)$ is the CDF of $Pr_\calD[c_i| T_i]$.
\end{proof}
\PropEquilibrium*
\begin{proof}
Let $c$ be the cost value such that $c=p(n(c),m(c))$ and fix a player $i$.   We show that the best response of player $i$ is to vote when $c_i\leq c$ and abstain otherwise irrespective of her type. 

If not, say $c_i > c$ and $i$ votes. Then the perceived  utility of player $i$ is given by 
\begin{align*}
     p(n(c), m(c)) - c_i  = c  - c_i< 0.  
\end{align*}

We note that the perceived pivotality $p$  only defines agents subjective belief on the value of individual vote and may not be from well defined probability space.  

Similarly,  when $c_i < c$, the best response strategy is to vote as voting gives strictly positive payoff to agent $i$.\footnote{Note that as agent considers his vote to be of value $p(.)$, this value is not realized if agent abstains giving her the payoff of zero.} Since the perceived pivotality is independent of player type, the voter strategies are the same.

Conversely, suppose there are two values $0< c_1<c_2\leq 1$ such that agents $i$ votes if $c_i \in [c_1, c_2]$ and abstains otherwise. Consider the following cases. 

\begin{enumerate}
    \item $c \in [c_1, c_2]$: In this case,  either there is an interval $\mathcal{I} := [c, c_2] $ such that agent $i$ votes (i.e. when $c_i \in \mathcal{I}$) but  in this case  voting is not a best response strategy for $i$  since $c_i < c$, or there is an interval $\mathcal{I}' = [0,c]$ such that voting is a best response strategy (when $c_i > c$). In either case, from continuity of support function we have that the probability measure on at least one of $I  $ and $\mathcal{I}'$ is nonzero. 
    \item $c \notin [c_1, c_2]$: We consider two subcases 
    \begin{itemize}
        \item $c < c_1$: In this case there is an interval $[c,c_1)$ such that the best response of agent $i$ is not to vote when  $c\geq c_i$. 
        \item  $c > c_2$: In this case there is an interval $[c_2,c)$ such that the best response of agent $i$ is not to vote when  $c\geq c_i$. 
    \end{itemize}
This completes the proof. 
\end{enumerate}
\end{proof}
\propTwo*
\begin{proof}
    For any $c\in [0,1]$, define $f(c):=p(n(c),m(c))$. Since the support functions are continuous, $n(c)$ and $m(c)$ are continuous (by Eqs.~\eqref{eq:n_c},\eqref{eq:m_c}). By continuity of $p$, the function $f$ is also continuous. Finally, every continuous function from $[0,1]$ to itself has a fixed point, due to intermediate value theorem applied to $f(x)-x$.
\end{proof}

\propThree*
\begin{proof}
    Assume towards a contradiction that there is a nontrivial equilibrium $\ol c=(c_N)_N$ with limit $c^*>\delta$ for some $\delta>0$. This means  that $c_N>q$ for all sufficiently large $N$. However, for any $N$, the number of active voters $n(c_N)$ is at least $(s_A(\delta)+s_B(\delta))N= \Omega(N)$, and thus $p(n,m)= O(\frac{1}{\sqrt{n}}) = O(\frac{1}{\sqrt{N}})$. This means that for sufficiently large $\ol{N}$, the pivot probability is $p(n(c_{\ol{N}}),m(c_{\ol{N}}))<\delta$. A contradiction.
\end{proof}
\propFour*
\begin{proof}
For any $\delta>0$, we should show that there is some $N_\delta$ and $c_{\delta}<\delta$ such that $c_\delta \in C^*(I,N_\delta)$.

W.l.o.g. $s_B(0)>s_A(0)$, and by the bounded derivatives, the margin $m(c)$ is at least some $\eps>0$ in some neighborhood $c\in [0,q]$.
Let $\ul{\delta}:=\min\{q,\delta\}$.

 By weakly vanishing pivotality, 
 $$p(n(\ul{\delta}),m(\ul{\delta}))\leq p(n(\ul{\delta}),\eps) = p(s_B(0)N,\eps)\rightarrow 0$$
 as $N$ grows. In particular there is some $N_\delta$ for which $p(s_B(0)N_\delta,\eps)<\ul{\delta}$. We argue that there must be an equilibrium in the range $[0,\ul{\delta}]$.

 Indeed, fix $N_\delta$ and let $f(x):=p(n(x),m(x))-x$. We know that $f$ is continuous, that $f(0)\geq 0$ and that $f(\ul{\delta}) = p(n(\ul{\delta}),m(\ul{\delta}))-\ul{\delta}< \ul{\delta}-\ul{\delta}=0$. From the intermediate value theorem, there must be some $c_\delta\in [0,\ul{\delta})$ for which $f(c_\delta)=0$, meaning $c_\delta=p(n(c_\delta),m(c_\delta))$. Thus $c_\delta\in C^*(I,N_\delta)$.
\end{proof}

\ProbOne*
\begin{proof}
We prove the result by looking at the best response of each type of voters. Let $u_i(p_i,c_i)$ be the utility of voter $i$ under perceived probability $p$ and cost of voting $c$. We will consider that two types ($A$ supporters and $B$ supporters) of voters have different perceived pivotalities $p$ and $p'$ respectively.   We will drop the subscript whenever it is clear from the context. 

First, note that since $n = n(c^+, c^+)$ we have that  either $c_A > c^+ > c_B$ or $c_B > c^+ >c_A$. That is, with same voting population as in the equilibrium,  either $A$'s supporters are more in number or $B$'s supporters are more than their number under the    equilibrium. 

\textbf{ Case 1 ($c_A  > c^+ > c_B > \widehat c$):} In this case we have 

$$s_A(c_A) > s_A(c^+) > s_B(c^+) > s_B (c_B) > s_A(\widehat c) = s_B(\widehat c). $$ Hence, there is a more participation from $A$ supporters compared to the equilibrium point $c^+$ and a correspondingly less participation from $B$ supporters than equilibrium point. 
This gives $ p  <  p^+ <  p'$.   

The utility for $A$ supporters  with cost  $c \in [c^+, c_A]$  (the shaded region on the right of $c^+$ in Figure \ref{fig:stabilityOne}) from voting is given by $$ u(p,c) = p - c <  p^+ - c = c^+ - c \leq  0.$$  Hence these voters are incentivized to abstain from election.  

The $B$ supporters in cost  $c \in [c_B, c^+]$, on the other hand, obtain strictly positive  utility by voting given by $$ u(p',c) = p'-c   >  p^+ - c \geq  p^+ - c^+ =0 $$ 

\textbf{Case 2 ($c_B > c^+ > c_A$):}  In this case, the $A$ voters' perceived pivotality $p$ is larger than that of equilibrium pivotality $p^+$ whereas $B$ supporters' perceived pivotality $p'$ is less that $p^+$ i.e., $ p' <  p^+ < p$. The utility of $A$ supporters in range $[c_A, c^+]$ (shaded region left of $c^+$ in Figure ~\ref{fig:stabilityOne}) from voting is 
$$u(p,c) = p - c > p - c^+ = p - p^+ > 0.$$ Similarly the utility obtained by $B$'s  supporters in cost range $c \in [c^+, c_B]$ by voting is given by 
$$ u(p',c) = p' - c < p^+ - c  = c^+ - c \leq  0.$$
Hence these $B$ supporters  are incetivized to vote.

    \begin{figure}
    \centering
\begin{tikzpicture}[scale=1.2]
\begin{axis}[
    xlabel={$c_A$},
    ylabel={$c_B$},
    xmin=0, xmax=100,
    ymin=0, ymax=100,
    legend pos=north west,
     axis y line*=left,
      axis x line*=bottom,
     ticks=none
]
    \node at (50,45) {$c^+$};
       \node at (20,15) {$\widehat c$};
    \draw[dotted] (0.07,0.25) -- (0.7,2);
   \addplot[color=red] coordinates {
    (0,0)(1,1)
    };
    \addplot[color=blue, dashed] coordinates {
    (0.3,0)(0.7,1)
    };
     \addplot[color=black, dashed] coordinates {
    (0.8,0)(0.2,1)
    };
  \node at (30,50) [module right arrow, color=red, scale=0.65] {$ \ A \ \ $};
  \node at (70,50) [module left arrow, color=red, scale=0.65] {$\ A \ \ $};
  \node at (50,30) [module up arrow, color=red, scale=0.65] {$ \ B \ \  $};
  \node at (50,70) [module down arrow, color=red, scale=0.65] {$\ B \ \ $};
 \node at (39,18) [rotate= 65] {$p = p^+$};
 \node at (73,18) [rotate=-56] {$p' = p^+$};
    \node at (50,50) [scale=1.5] {$\star$};
     \end{axis} 
\end{tikzpicture}
   \caption{A schematic illustation of stability of the equilibrium point $c^+$. When $c' > c^+$ is an equilibrium estimate, the $B$  supporters from  right shaded region are incentivized to participate  whereas under $c"<c^+$ the $B$ supporters from left shaded region are incetivized to abstain. \rmr{move to appendix}      }
    \label{fig:stabilityTwo}
\end{figure} 

\end{proof}

\section{A Detailed Example }
\label{sec:example}
\exOne*

 \subsection{Analysis of Example \ref{ex:first} under fully rational models}

\emph{Binomial PPM:}  The expected number of active voters for a given cost $c$ is given as $n(c) = N(s_A(c)+s_B(c)) = N(1+c)/2$. As for the winning  margin, we have 
$m(c) = \frac{|0.3-c/2|}{0.5+c/2}=\frac{|0.6-c|}{1+c}.$
  The equilibrium $c^*$ satisfies 
\begin{align*}
c^*&= \ppm(n(c^*),m(c^*)) = \ppm\Bigg(N\frac{1+c^*}{2},\frac{0.6-c^*}{1+c^*}\Bigg)\\
&\cong \frac{1}{\sqrt{\pi N(1+c^*)}} \Bigg( 1 - \Big ( \frac{0.6-c^*}{1+c^*} \Big)^2 \Bigg )^{N(1+c^*)/4}  \leq  \frac{1}{\sqrt{\pi N }}
\end{align*}

For large value of $N$, $c^*$ is close to $0$ and hence the margin remains close to $0.6$ in favor of the unpopular candidate $B$,   increasing the winning probability of $B$. Note that this equilibrium point indeed satisfies the Jury theorem as given above. 
 

 \subsection{Analysis of Example \ref{ex:first} under proposed tie-sensitive PPM model}
Consider the proposed tie-sensitive model  $\ppm(n,m)=\min\{1,1/m\sqrt n\}$ and observe that
$c^*  = \ppm(n(c^*),m(c^*)) = \min(1, \frac{1}{m(c^*)\sqrt{n(c^*)}})  
 = \min (1, \frac{\sqrt{1+c^*}}{|0.6-c^*|} \sqrt{\frac{2}{N}}  ).$ 
For large values of $N$ we obtain  three non-negative solutions for $c^*$; $ 0 <c^0 <  c^- < 0.6 <  c^+ \leq 1$. For instance. for $N = 10^4$ we have. $(c^0, c^- , c^+ )$ = $(0.02489,0.5689, 0.62823)  $  and for $N=10^6$ it is $(0.002369, 0.5970, 0.602964)$. See Figure \ref{fig:pivotProb}.


 \noindent  \textbf{Trivial Equilibrium:} We first track the equilibrium around $0$; i.e. $c^0$. Similar to the strongly vanishing Binomial  PPM model, this model also supports  an equilibrium solution that goes to $0$ with rate $N^{-1/2}$. Furthermore, a trivial solution always exists and candidate $B$ wins in a trivial equilibrium satisfying the Jury theorem.

\noindent  \textbf{Non-trivial Equilibria:} 
We will consider the equilibrium $c^+$ i.e. the one in  which popular candidate (candidate~A) wins with large probability. The analysis of $c^-$ is analogous.   
Let $c^+:=0.6 + \eps$ and note that for large $N$, the winning margin  $m(c^+)$ is approximately  $\eps/1.6$ and the number of active voters is   
$n(c^+)= N(1+c^+)/2=N(1.6+\eps)/2=0.8N+0.5 N\eps$. Hence 
\begin{align*}
     c^+  &= \eps +0.6  =  \ppm(0.8N+N\eps/2,\eps/1.6) \\ & \cong  \min\{1,\frac{1.6}{\eps\sqrt{0.8N+ 0.5 N\eps}}\}
\end{align*}  

 The minimum on the right side of the above equation is less than $1$ for large enough $N$. We have  \begin{equation}(\eps+0.6)\eps\sqrt{0.8N+0.5 N\eps} = 1.6 \end{equation}
The following sequence of inequalities   
$$0.5\eps\sqrt N< 0.6\eps\sqrt{0.8N}< (\eps +0.6)\eps \sqrt{0.8N+ 0.5N\eps} = 1.6 $$ 
imply   $\eps  < \frac{3.2}{\sqrt{N}}$. This gives an upper bound on the rate at which $\varepsilon$ decreases as a function of population size $N$. Similarly, we lower bound the value of $\varepsilon$ using
\begin{align*}
    1.6 = (\eps +0.6)\eps \sqrt{0.8N+ 0.5N\eps} \leq  \varepsilon \sqrt{N} 
\end{align*} to $ 1.6/\sqrt{N}$.
The last inequality above follows from the fact that $c^+ = 0.6 + \varepsilon \leq 1$ i.e. $\varepsilon \leq 0.4$.
 Together we  have \begin{equation}
 \frac{1.6}{\sqrt{N}}  \leq  \varepsilon \leq  \frac{3.2}{\sqrt{N}}
 \label{eq:boundepsilon}
\end{equation}
For large value of $N$, Eq. \eqref{eq:boundepsilon} above gives a tight characterization of the equilibrium threshold point $c^+$.  This, consequently, gives a reasonably tight bound on the probability that a random voter will vote for  popular candidate $A$ is given by \begin{equation*}
    \Pr(A) := \frac{s_A(c^+)}{s_A(c^+) + s_B(c^+)} = \frac{0.2+c^+}{1+c^+}. 
\end{equation*} 
A straightforward calculation, gives the following bounds on $P(A)$.
\begin{equation}  0.5 + \frac{1}{2(\sqrt{N}+1)} < P(A) < 0.5 + \frac{1}{2+ \sqrt{N}} \end{equation}
Let  $\delta_1 = \frac{1}{2(\sqrt{N}+1)}$ and $\delta_2 = \frac{1}{2+ \sqrt{N}}$. 

Using this, we now analyze the win  probability of candidate $A$. For this, we will approximate the Binomial random variable by the appropriate gaussian random variable. Let,  $\mu = n(c^+)\Pr(A)$ and  $\sigma = \sqrt{n(c^+) \Pr(A)(1 - \Pr(A))}$. Note that $\Pr(\text{A wins}) \simeq     \Pr_{X \sim \mathcal{N}(\mu, \sigma^2) } (X \geq \frac{n(c^+)}{2})$. 

Using the bounds above  we see that
\begin{align*}\sigma &> \sqrt{0.8 N (0.25 - \varepsilon_2)^2} =  \sqrt{0.8 N (0.25 - 1/(2 + \sqrt{N})^2)}  \\ & \approx  \sqrt{0.2 N} > 0.4 \sqrt{N}\end{align*} 
On the other hand \begin{align*}
    |n(c^+)/2 - \mu | &\leq  n(c^+)  \varepsilon_2 =   0.8\sqrt{N}   <  2 \sigma.
\end{align*}
This implies that the difference between the mean and the threshold for $A$'s victory is upper bounded by two times standard deviations.  Hence the candidate $B$ stands at least a 0.6 percentage chance of winning the election. 

 \begin{figure}[ht!]
 \centering 
 \begin{tikzpicture}[scale=1.3]
\def\a{40}
\def\bx{-20}
\def\by{-13}
\draw[blue,thick] (0.5*\a+\bx, 0.4*\a+\by) -- (0.7*\a+\bx, 0.4*\a+\by);
\draw[red,double] (0.5*\a+\bx, 0.35*\a+\by) -- (0.7*\a+\bx, 0.45*\a+\by);
\draw[thick,->] (0,0.3) -- (8,0.3);
\draw[dashed,thick] (0.628*\a+\bx,0.25) -- (0.628*\a+\bx,5);
\draw[dashed,thick] (0.569*\a+\bx,0.25) -- (0.569*\a+\bx,5);
\draw[dotted,thick] (0.628*\a+\bx,  0.414*\a+\by) -- (0, 0.414*\a+\by);
\draw[dashed,thick] (0.51*\a+\bx,0.25) -- (0.51*\a+\bx,5);

\node at (0.6*\a+\bx, 0.4*\a+\by) {$\circ$};

\node at (0.628*\a+\bx,-0) {$0.628$};
\node at (0.628*\a+\bx,-0.5) {$c^+$};
\node at (0.569*\a+\bx,-0) {$0.569$};
\node at (0.51*\a+\bx,-0) {$0.025$};

\node at (8,0) {$c$};

\node at (0.569*\a+\bx,-0.5) {$c^-$};
\node at (0.6*\a+\bx,-0) {$0.6$};
\node at (0.6*\a+\bx,-0.5) {$\hat c$};
\node at (0.51*\a+\bx,-0.5) {$c^0$};
\node at (-0.4, 0.4*\a+\by) {$0.4$};
\node at (-0.5, 0.35*\a+\by) {$0.35$};
\node at (-0.6, 0.414*\a+\by) {$0.414$};

\node[blue] at (0.7*\a+\bx, 0.4*\a+\by +0.3) {$s_B = 0.4$};

\node[red] at (0.7*\a+\bx, 0.45*\a+\by+0.3) {$s_A = 0.1 + \frac{c}{2}$};

\draw[|<->|] (0.628*\a+\bx+0.2, 0.4*\a+\by+0.01) -- (0.628*\a+\bx+0.2, 0.414*\a+\by );
\draw[<->|] (0.628*\a+\bx+0.1, 0.5) -- (0.628*\a+\bx+0.1, 0.4*\a+\by-0.01 );
\draw[<->|] (0.628*\a+\bx-0.1, 0.5) -- (0.628*\a+\bx-0.1, 0.414*\a+\by );
\node[rotate=90] at (0.628*\a+\bx+0.4,1.5) {$\sim Pr({vote}~B)$};
\node[rotate=90] at (0.628*\a+\bx-0.4,1.6) {$\sim Pr({vote}~A)$};
\node at (0.628*\a+\bx+0.6, 0.4*\a+\by+0.3) {$m'$};





\end{tikzpicture}
\caption{The pivot point $\hat c$ is marked by a circle at the intersection of the support functions, with the two non-trivial equilibria for $N=10000$ on its sides (dashed lines). For the upper equilibrium $c^+$, the probability of a random voter to vote $A$ or $B$ is proportional to $s_A(c^+)=0.414$ and $s_B(c^+)=0.4$, respectively. Thus the margin is $m'$ divided by the fraction of active voters $s_A(c^+)+s_B(c^+)=0.814$.}
\label{fig:pivotProb}
\end{figure}
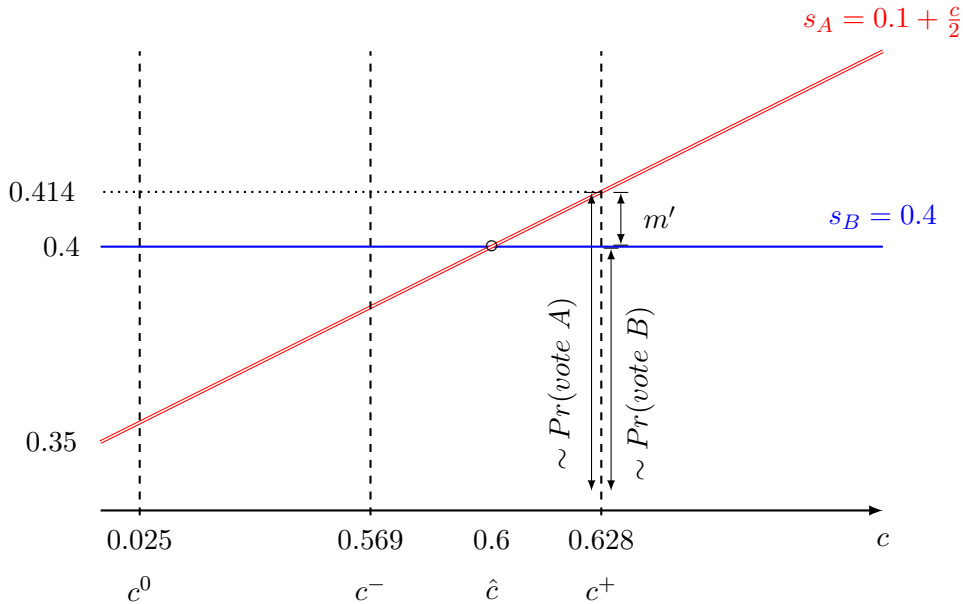


\section{ Additional Numerical results}
In this section, we empirically evaluate the theoretical results.  We first show in Figure \ref{fig:lastbutone}   that the number of active voters in large elections is  a constant fraction of the population size.  We consider an election instance from Section \ref{sec:example}  and observe that for a large value of $N$ both the equilibria induce (almost) a constant fraction of voters to vote. This also implies that for a large value of $N$ both equilibrium points are nearby, favouring different candidates whereas not guaranteeing any candidate a certain victory.

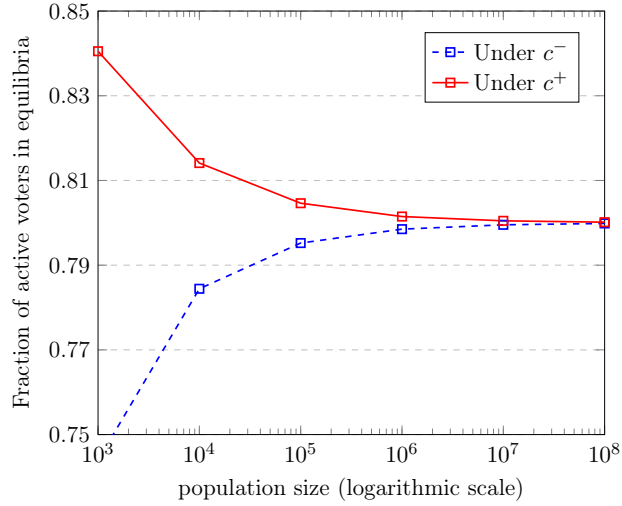
\begin{figure}[ht!]
\centering 
\scalebox{0.8}{
\begin{tikzpicture}
\begin{axis}[
xmode = log,
    log ticks with fixed point,
    xlabel={population size (logarithmic scale)},
    ylabel={Fraction of active voters in equilibria},
    xmin=1000, xmax=10^8,
    ymin=0.75, ymax=0.85,
    xtick={10^3, 10^4, 10^5, 10^6,10^7,10^8},
    xticklabels = {$10^3$,$10^4$, $10^5$, $10^{6}$,$10^7$,$10^8$},
    ytick={0.75, 0.77, 0.79, 
    0.81,0.83,0.85},
    yticklabels = {0.75, 0.77, 0.79, 
    0.81,0.83,0.85},
    legend pos=north west,
    ymajorgrids=true,
    grid style=dashed,
legend style={
at={(0,0)},
anchor=north east,at={(axis description cs:0.95,0.95)}}]
\addplot[thick, 
    color=blue,
    mark=square, mark options={solid},  dashed
    ]
    coordinates {
   (10^3, 743.960609/10^3)(10^4, 7844.27549/10^4)(10^5, 79522.3891/10^5)(10^6, 798503.213/10^6)(10^7, 7995279.92/10^7) (10^8, 79985086.9/10^8)
    };
\addplot[ thick,
    color=red,
    mark options={solid},
    mark=square
    ]
    coordinates {
   (10^3, 840.507978/10^3) (10^4, 8141.13369/10^4) (10^5, 80462.9039/10^5)(10^6, 801482.012/10^6)(10^7, 8004705.28/10^7)(10^8,  80014898.3/10^8)
    };
\legend{ Under $c^{-}$, Under $c^+$}
\end{axis}
\end{tikzpicture}
}

\caption{ Number of active voters is a constant fraction of population size under both the equilibrium points.  }
\label{fig:lastbutone}
\end{figure}

Next, in Figure \ref{fig:last} we observe that for the {\em weaker} dependence on the margin of victory i.e. smaller values of $\alpha$ induce smaller equilibrium cost values. For instance, for $N = 10^3$  the equilibrium is  $c^+ =0.8$ for  $\alpha =1.5$ whereas the equilibrium cost is $c^+ = 0.605$ when $\alpha = 0.5$.  This implies that for fully rational models, when $\alpha =0$ our models predicts that only equilibrium is the trivial equilibrium. 

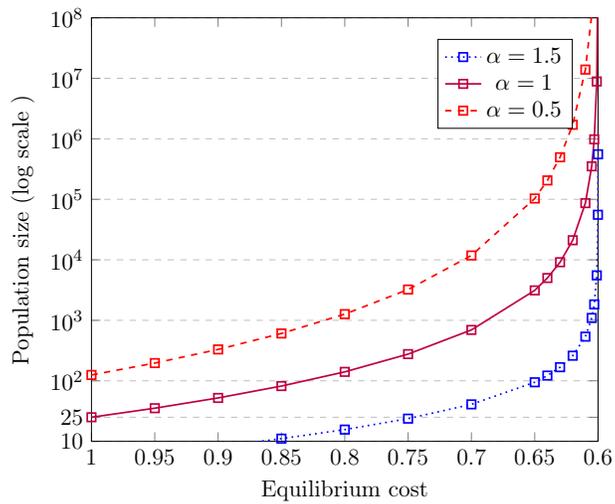
\begin{figure}[ht!]
\centering
\scalebox{0.8}{
\begin{tikzpicture}
\begin{axis}[
 ymode = log,
    log ticks with fixed point,
    xlabel={Equilibrium cost },
    ylabel={Population size (log scale )},
    xmin= -1, xmax= -0.6,
    ymin=10, ymax=10^8,
    xtick={ -0.6, -0.65, -0.7, -0.75, -0.8, -0.85, -0.9, -0.95, -1},
    xticklabels = {0.6, 0.65, 0.7, 0.75, 0.8, 0.85, 0.9, 0.95, 1},
    ytick={10, 25,10^2,10^3, 10^4, 10^5, 
    10^6,10^7,10^8},
    yticklabels = {10, 25,$10^2$,$10^3$, $10^4$, $10^5$, 
    $10^6$,$10^7$,$10^8$, $10^9$, $10^10$},
    legend pos=north west,
    ymajorgrids=true,
    grid style=dashed,
legend style={
at={(0,0)},
anchor=north east,at={(axis description cs:0.95,0.95)}}]
\addplot[thick, dotted, 
    color=blue,
    mark=square, mark options={solid}
    ]
    coordinates {(-1, 5.0000000000000009)
(-0.95, 6.3316185199841737)
(-0.9, 8.2304526748971174)
(-0.85, 11.072664359861593)
(-0.8, 15.625)
(-0.75, 23.703703703703717)
(-0.7, 40.816326530612308)
(-0.65, 94.674556213017652)
(-0.64, 122.07031249999987)
(-0.63, 167.96842193667575)
(-0.62, 260.14568158168549)
(-0.61, 537.48992206396076)
(-0.605, 1092.8215285841366)
(-0.603, 1833.4713020488789)
(-0.601, 5537.0832306671919)
(-0.6001, 55537.041665644203)
(-0.60001, 555537.03749635024)
  
    };
        \addplot[thick, 
    color=purple,
    mark=square, mark options={solid}
    ]
    coordinates {
(-1, 25.000000000000014)
(-0.95, 35.276160325626122)
(-0.9, 52.126200274348413)
(-0.85, 81.937716262975798)
(-0.8, 140.62500000000006)
(-0.75, 276.54320987654353)
(-0.7, 693.87755102041001)
(-0.65, 3124.2603550295798)
(-0.64, 5004.8828124999909)
(-0.63, 9126.2842585593753)
(-0.62, 21071.800208116507)
(-0.61, 86535.877452297616)
(-0.605, 350795.71067551529)
(-0.603, 979684.83239481959)
(-0.601, 8864870.2522991505)
(-0.6001, 888648203.69207084)
(-0.60001, 88886481535.871216)   };

        \addplot[thick, dashed, 
    color=red,
    mark=square, mark options={solid}
    ]
    coordinates {(-1, 125.00000000000009)
(-0.95, 196.53860752848848)
(-0.9, 330.13260173753986)
(-0.85, 606.33910034602093)
(-0.8, 1265.6250000000007)
(-0.75, 3226.3374485596769)
(-0.7, 11795.918367346987)
(-0.65, 103100.59171597607)
(-0.64, 205200.19531249948)
(-0.63, 495861.4447150589)
(-0.62, 1706815.8168574357)
(-0.61, 13932276.269819904)
(-0.605, 112605423.12684281)
(-0.603, 523478262.10965091)
(-0.601, 14192657273.932503)
(-0.6001, 14219259907278.391)
(-0.60001, 14221925932127762.0)  };
\legend{ $ \alpha = 1.5$,$\alpha =1$, $\alpha = 0.5 $}
\end{axis}
\end{tikzpicture}
}
\caption{The size of the population that induces equilibrium threshold as a given cost value.  }
\label{fig:last}
\end{figure}

\subsection{Comparing $c^-$ to $c^+$}

\begin{figure}[ht!]
\centering 
 \resizebox{0.6\columnwidth}{0.5 \columnwidth}{
\begin{tikzpicture}
\begin{axis}[
 xmode=log,
    log ticks with fixed point,
    xlabel={population size (logarithmic scale)},
    ylabel={Equilibrium win probability},
    xmin=0, xmax=10^8,
    ymin=0.9, ymax=1,
    label style={font=\Large},
                    tick label style={font=\Large}, 
    xtick={10^3, 10^4, 10^5, 10^6,10^7,10^8},
    xticklabels = {$10^3$,$10^4$, $10^5$, $10^{6}$,$10^7$,$10^8$},
    ytick={0.88,0.9,0.92,0.94,0.96,0.98,1},
    legend pos=north west,
    ymajorgrids=true,
    grid style=dashed,
legend style={
at={(0,0)},
anchor=north east,at={(axis description cs:0.95,0.95)}}]
\addplot[ thick, 
    color=red,
    mark=square, mark options={solid},  dashed
    ]
    coordinates {
  
   (10^3, 0.90810929)(10^4, 0.93644931)(10^5, 0.95076779)(10^6, 0.95171826) (10^7, 0.95206293)(10^8, 0.95216058)
  
    };
\addplot[ thick,
    color=blue,
    mark options={solid},
    mark=square
    ]
    coordinates {
    (10^3, 0.98266767)(10^4, 0.96161482)(10^5, 0.95419222)(10^6,  0.95247038)(10^7,  0.95230315)(10^8,0.95223574) 
    };
\legend{Win probability of $A$ under $c^{+}$,Win probability of $B$ under $c^-$}
\end{axis}
\end{tikzpicture}
}
\caption{ Win probability for different values of $N$ under respective induced equilibria.  The equilibrium win probability for a popular candidate $A$ increases with $N$ whereas it decreases for the unpopular candidate $B$. \rmr{maybe add in appendix a similar figure  for $n^*$ as a function of N}}
\label{fig:win_per_size2}
\end{figure}
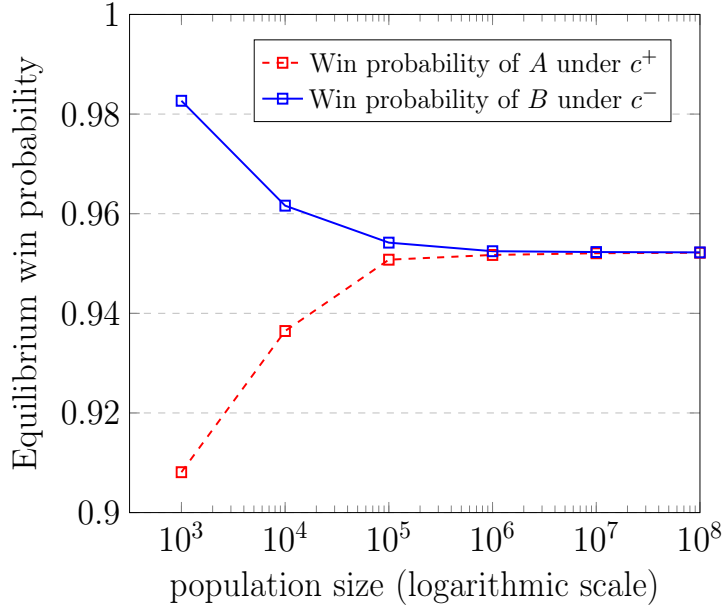

Fig.~\ref{fig:win_per_size2} shows empirically that
\begin{enumerate}
    \item The locally popular candidate at $c^-$ (B) has decreasing winning probabilities;
    \item The locally popular candidate at $c^+$ (A) has increasing winning probabilities;
    \item The winning probability of B at $c^-$ is always larger.
\end{enumerate}
At least for the third observation we can provide a theoretical justification.

\begin{proposition}
For a linear-support, tie-sensitive PPM, and any $N$, it holds that the winning probability of $B$ at $c^-$ is strictly greater than the winning probability of $A$ at $c^+$.
\end{proposition}
\begin{proof}
    Since $c^-<c^+$, there are strictly less active voters (of both types) in $c^-$, thus $n(c^-)<n(c^+)$.

    Since both points are equilibrium point, we know that voters in $c^+$ participate iff their perceived pivotality is at least $c^+$. That is, 
    $$p(n(c^+),m(c^+))=c^+.$$

    Similarly, 
    $$p(n(c^-),m(c^-))=c^-.$$
    Now, $p$ is decreasing in both arguments (for any PPM), thus
    $$p(n(c^-),m(c^-)) > p(n(c^+),m(c^-)).$$
    Together, 
    $$p(n(c^+),m(c^+))=c^+>c^-=p(n(c^-),m(c^-)) > p(n(c^+),m(c^-)). $$
    Since $p$ is decreasing in $m$, this means $m(c^+)<m(c^-)$. 
\end{proof}


\newpage 
\begin{thebibliography}{}
\bibitem[Ald93]{Aldrich}
John~H. Aldrich.
\newblock Rational choice and turnout.
\newblock {\em American Journal of Political Science}, 37(1):246--278, 1993.

\bibitem[Bla00]{Blais}
André Blais.
\newblock {\em To Vote or Not to Vote: The Merits and Limits of Rational Choice Theory}.
\newblock University of Pittsburgh Press, 2000.

\bibitem[Dow57]{Downs57}
Anthony Downs.
\newblock An economic theory of political action in a democracy.
\newblock {\em Journal of Political Economy}, 65(2):135--150, 1957.

\bibitem[EGK07]{edlin2007voting}
Aaron Edlin, Andrew Gelman, and Noah Kaplan.
\newblock Voting as a rational choice: Why and how people vote to improve the well-being of others.
\newblock {\em Rationality and society}, 19(3):293--314, 2007.

\bibitem[FF75]{Ferejohn75}
John~A. Ferejohn and Morris~P. Fiorina.
\newblock Closeness counts only in horseshoes and dancing.
\newblock {\em The American Political Science Review}, 69(3):920--925, 1975.

\bibitem[FLMG19]{fairstein2019modeling}
Roy Fairstein, Adam Lauz, Reshef Meir, and Kobi Gal.
\newblock Modeling people's voting behavior with poll information.
\newblock In {\em Proceedings of the 18th International Conference on Autonomous Agents and MultiAgent Systems}, pages 1422--1430, 2019.

\bibitem[FP96]{fedderson12}
{Timothy J.} Feddersen and Wolfgang Pesendorfer.
\newblock The swing voter's curse.
\newblock {\em American Economic Review}, 86(3):408--424, 1996.

\bibitem[KS85]{katz1985network}
Michael~L Katz and Carl Shapiro.
\newblock Network externalities, competition, and compatibility.
\newblock {\em The American economic review}, 75(3):424--440, 1985.

\bibitem[Mei18]{2018Meir}
Reshef Meir.
\newblock {\em Strategic Voting}.
\newblock Synthesis Lectures on Artificial Intelligence and Machine Learning. Morgan {\&} Claypool Publishers, 2018.

\bibitem[Mer81]{Merrill1981}
Samuel Merrill.
\newblock Strategic decisions under one-stage multi-candidate voting systems.
\newblock {\em Public Choice}, 36(1):115--134, 1981.

\bibitem[MGT20]{meir2020strategic}
Reshef Meir, Kobi Gal, and Maor Tal.
\newblock Strategic voting in the lab: compromise and leader bias behavior.
\newblock {\em Autonomous Agents and Multi-Agent Systems}, 34(1):1--37, 2020.

\bibitem[MLR14]{Meir14}
Reshef Meir, Omer Lev, and Jeffrey~S. Rosenschein.
\newblock A local-dominance theory of voting equilibria.
\newblock In {\em Proceedings of the Fifteenth ACM Conference on Economics and Computation}, EC '14, page 313–330, New York, NY, USA, 2014. Association for Computing Machinery.

\bibitem[MW93]{MyersonWeber}
Roger~B. Myerson and Robert~J. Weber.
\newblock A theory of voting equilibria.
\newblock {\em The American Political Science Review}, 87(1):102--114, 1993.

\bibitem[Mye98]{Myerson1998}
Roger~B. Myerson.
\newblock Population uncertainty and poisson games.
\newblock {\em International Journal of Game Theory}, 27(3):375--392, 1998.

\bibitem[Mye02]{MYERSON2002219}
Roger~B Myerson.
\newblock Comparison of scoring rules in poisson voting games.
\newblock {\em Journal of Economic Theory}, 103(1):219--251, 2002.

\bibitem[OG84]{Owen1984}
Guillermo Owen and Bernard Grofman.
\newblock To vote or not to vote: The paradox of nonvoting.
\newblock {\em Public Choice}, 42(3):311--325, 1984.

\bibitem[OR03]{OSBORNE2003434}
Martin~J. Osborne and Ariel Rubinstein.
\newblock Sampling equilibrium, with an application to strategic voting.
\newblock {\em Games and Economic Behavior}, 45(2):434--441, 2003.
\newblock Special Issue in Honor of Robert W. Rosenthal.

\bibitem[RO68]{riker68}
William~H. Riker and Peter~C. Ordeshook.
\newblock A theory of the calculus of voting.
\newblock {\em The American Political Science Review}, 62(1):25--42, 1968.

\bibitem[Wik22]{wiki:Voter_turnout}
Wikipedia.
\newblock {Voter turnout in United States presidential elections} --- {W}ikipedia{,} the free encyclopedia.
\newblock \url{http://en.wikipedia.org/w/index.php?title=Voter\%20turnout\%20in\%20United\%20States\%20presidential\%20elections&oldid=1116061047}, 2022.
\newblock [Online; accessed 28-October-2022].

\end{thebibliography}
\end{document}